\documentclass[preprint,12pt]{elsarticle}

\usepackage{amssymb}
 \usepackage{amsthm}

\journal{Discrete Applied Mathematics}
\usepackage[utf8]{inputenc}
\usepackage{enumerate}

\usepackage{hyperref}

\newcommand{\siom}{\sigma_{\om_i}}
\newcommand{\siomm}{\sigma_{\om_{i-1}}}
\newcommand{\om}{\omega}
\newcommand{\chid}{\chi'^d_{\sum}}

\usepackage{color}
\usepackage{graphicx}
\usepackage{amsmath}

\usepackage{subfig}

\makeatletter
\newtheorem*{rep@theorem}{\rep@title}
\newcommand{\newreptheorem}[2]{%
	\newenvironment{rep#1}[1]{%
		\def\rep@title{#2 \ref{##1}}%
		\begin{rep@theorem}}%
		{\end{rep@theorem}}}
\makeatother

\newtheorem{theorem}{Theorem}
\newtheorem{lemma}[theorem]{Lemma}

\newtheorem{obs}[theorem]{Observation}

\newtheorem{corollary}[theorem]{Corollary}

\newtheorem{claim}{Claim}

\newtheorem{remark}[theorem]{Remark}
\newreptheorem{theorem}{Theorem}

\begin{document}
\begin{frontmatter}

\title{Neighbour sum distinguishing edge-weightings with local constraints}


\author[1,2,3]{Antoine Dailly}
\ead{antoine.dailly@im.unam.mx}

\author[4]{El\.zbieta Sidorowicz\corref{cor1}}
\ead{e.sidorowicz@wmie.uz.zgora.pl}

\address[1]{Instituto de Matem\'aticas, UNAM Juriquilla, 76230 Quer\'etaro, Mexico}

\address[2]{G-SCOP, Universit\'e Grenoble Alpes, CNRS, Grenoble, France}

\address[3]{LIMOS, CNRS UMR 6158, Université Clermont Auvergne, Aubi\`ere, France}

\address[4]{Institute of Mathematics, University of Zielona G\'ora  ul. prof. Z. Szafrana 4a, 65-516 Zielona G\'ora, Poland}

\cortext[cor1]{Corresponding author}

\begin{abstract}
A $k$-edge-weighting of $G$ is a mapping $\omega:E(G)\longrightarrow \{1,\ldots,k\}$. The edge-weighting  of $G$ naturally induces a vertex-colouring $\sigma_{\omega}:V(G)\longrightarrow \mathbb{N}$ given by
$\sigma_{\omega}(v)=\sum_{u\in N_G(v)}\om(vu)$ for every $v\in V(G)$. The edge-weighting $\om$ is neighbour sum distinguishing if it yields a proper vertex-colouring  $\sigma_{\om}$, \emph{i.e.},  $\sigma_{\om}(u)\neq \sigma_{\om}(v)$ for every edge $uv$ of $G$.
		
We investigate a neighbour sum distinguishing edge-weighting with  local constraints, namely, we assume that the set of  edges incident to a vertex of large  degree  is not monochromatic. A graph is nice if it has no components isomorphic to $K_2$. We prove that every nice graph with maximum degree at most~5 admits a neighbour sum distinguishing $(\Delta(G)+2)$-edge-weighting such that all the vertices of degree at least~2 are incident with at least two edges of different weights. Furthermore, we prove that every nice graph admits a neighbour sum distinguishing $7$-edge-weighting such that all the vertices of degree at least~6 are incident with at least two edges of different weights. Finally, we show that nice bipartite graphs admit a neighbour sum distinguishing $6$-edge-weighting such that all the vertices of degree at least~2 are incident with at least two edges of different weights.

\end{abstract}

\begin{keyword}

1-2-3 Conjecture\sep neighbour sum distinguishing edge weighting\sep neighbour sum distinguishing edge colouring 




\end{keyword}

\end{frontmatter}

\section{Introduction}
	Let $G$ be a graph and $k\in \mathbb N^*$. Every \emph{$k$-edge-weighting}, \emph{i.e.}, a mapping $\omega:E(G)\longrightarrow \{1,\ldots,k\}$, induces a vertex-colouring $\sigma_{\om}:V(G)\longrightarrow \mathbb{N}$, where $\sigma_{\omega}(v)=\sum_{u\in N_G(v)}\om(vu)$. In natural language, assigning weights to edges allows us to obtain a vertex-colouring by assigning to each vertex the sum of the weights of its incident edges. We say that the edge-weighting $\om$ \emph{distinguishes} vertices $v,w \in V(G)$ if $\sigma_{\om}(v)\neq \sigma_{\om}(w)$, and that $\om$ is \emph{neighbour sum distinguishing} (or simply \emph{distinguishing}) if it distinguishes every pair of adjacent vertices.
	
	Hence, a \emph{neighbour sum distinguishing $k$-edge-weighting} is a mapping $\omega:E(G)\longrightarrow \{1,\ldots,k\}$ that is distinguishing, \emph{i.e.}, the induced vertex-colouring $\sigma_{\om}$ is proper. 
	Observe that $G$ always admits such a neighbour sum distinguishing edge-weighting, unless it includes $K_2$ as a  component: assign a different power of~2 to every edge, thus ensuring that every vertex will get a different sum; however, in a $K_2$ component, the two vertices cannot be distinguished. Hence, we call $G$ \emph{nice} whenever it has no such component.
	
	In 2004 Karo\'nski et al. \cite{KaLu04} posed the conjecture, called the 1-2-3 Conjecture, that  asks whether every nice graph  admits  a $3$-edge-weighting that is neighbour sum distinguishing. The 1-2-3 Conjecture inspired  a lot of  studies on the original conjecture and variants of it. For more information on that topic, we refer the reader to the survey by Seamone \cite{Se12}. The best result towards the 1-2-3 Conjecture
	is due to Kalkowski et al. \cite{KaKa11}, who proved that every nice graph  admits a neighbour sum distinguishing $5$-edge-weighting.
	The conjecture cannot be pushed further down, since there are graphs that require three weights, as an example, see cycles or complete graphs. It was proved by Dudek and Wajc \cite{DuWa11} that  deciding whether there is a neighbour sum distinguishing $2$-edge-weighting for a given  graph  is NP-complete in general, while Thomassen, Wu and Zhang \cite{ToWu16} showed that the same problem is polynomial-time solvable in the family of bipartite graphs.
	Recently Przyby{\l}o \cite{Pr20} proved that  every $d$-regular graph $(d\ge 2)$ admits a neighbour sum distinguishing $4$-edge-weighting  and that the 1-2-3 Conjecture is true for $d$-regular graphs with $d\ge 10^8$. 
	
	In the version of the neighbour sum distinguishing edge-weighting, introduced by Karo\'nski et al. \cite{KaLu04}, the edges incident with a vertex may have the same weight. On the other hand, Flandrin et al. \cite{FlMa13} introduced the version of the  edge-weighting, called a neighbour sum  distinguishing $k$-edge-colouring, which distinguishes vertices and in which adjacent edges must have different weights.
	A {\it $k$-edge-colouring} of $G$ is a mapping $\omega:E(G)\longrightarrow \{1,\ldots,k\}$ such that $\om(e_1)\neq \om(e_2)$ for every two adjacent edges $e_1,e_2\in E(G)$. If the  $k$-edge-colouring $\omega$ is distinguishing, 
	then we call such a colouring a \emph{neighbour sum  distinguishing $k$-edge-colouring}.  The smallest value $k$ for which $G$ admits  a neighbour sum  distinguishing $k$-edge-colouring is denoted by $\chi'_{\sum}(G)$. Flandrin et al. conjectured that every nice graph $G$ except $C_5$ verifies $\chi'_{\sum}(G) \leq \Delta(G)+2$.
	Wang and Yan \cite{WaYa2014} proved that $\chi'_{\sum}(G)\le  \left\lceil (10\Delta(G)+2)/3\right\rceil$ when $\Delta(G)\ge 18$. It is known that $\chi'_{\sum}(G)\le  2\Delta(G)+{\rm col}(G)-1$ \cite{Pr15} and  $\chi'_{\sum}(G)\le  \Delta(G)+3{\rm col}(G)-4$ \cite{PrzWo15}, where ${\rm col}(G)$ denotes the  {\it colouring number of} $G$, \emph{i.e.} the smallest integer $k$ such that $G$ has a vertex ordering in which each vertex is preceded by fewer than $k$ of its neighbours. Recently, Przyby{\l}o \cite{Pr19b} proved that $\chi'_{\sum}(G)\le\Delta+O(\sqrt{\Delta})$, where $\Delta=\Delta(G)$.
	
	In a previous work~\cite{DaDuPeSi}, the authors proposed a generalization of both conjectures, by introducting the notion of \emph{neighbour sum  distinguishing relaxed edge-colouring}. The idea is to have a continuum of parameters from the neighbour sum distinguishing edge-weighting of the 1-2-3 Conjecture to the neighbour sum  distinguishing edge-colouring of the proper variant, by allowing each vertex to be incident with a limited number of edges of the same colour.
	
	More formally, a $d$-\emph{relaxed $k$-edge-colouring} is a mapping $\omega:E(G)\longrightarrow \{1,\ldots,k\}$ such that each  monochromatic set of edges induces a subgraph with maximum degree at most~$d$. If a $d$-relaxed $k$-edge-colouring $\omega$  is distinguishing, then it is called a \emph{neighbour sum  distinguishing $d$-relaxed $k$-edge-colouring}. By $\chid(G)$, we denote the smallest value $k$ for which $G$ admits a neighbour sum  distinguishing $d$-relaxed $k$-edge-colouring. Hence, $\chi'^{1}_{\sum}(G)=\chi'_{\sum}(G)$, and $\chi'^{\Delta(G)}_{\sum}(G)$ corresponds to the 1-2-3 Conjecture. The general conjecture stated in~\cite{DaDuPeSi} is that every nice graph except $C_5$ verifies $\chi'^{1}_{\sum}(G) \leq \left\lceil \frac{\Delta(G)}{d} \right\rceil+2$.
	
	One way of studying graph parameters is by bounding the maximum degree of the graph. While graphs of maximum degree~2 (forests of paths and cycles) are generally easy, the case of subcubic graphs is often not trivial. Indeed, while the 1-2-3 Conjecture holds for subcubic graphs~\cite{KaLu04}, the proper variant is still open, with the best bound as of now being~6~\cite{HuWa17} (the conjecture states that the bound should be~5). In~\cite{DaDuPeSi}, the 2-relaxed case for the general conjecture was settled for subcubic graphs, and in fact a more general result was proved: every nice subcubic graph with no component isomorphic to $C_5$ admits a neighbour sum  distinguishing $2$-relaxed $4$-edge-colouring such that every vertex of degree~2 is incident with edges coloured differently.
	
	Inspired by this result, we consider edge-weightings allowing a vertex to be incident with edges having the same weight, in a limited way. We require that a vertex of large enough degree is incident with at least two edges of different weights. Such a version is, on the one hand, stronger than the classical edge-weighting, while, on the other hand, it is weaker than the edge-colouring. Indeed, observe that if $G$ admits a neighbour sum  distinguishing $k$-edge-weighting such that every vertex of degree at least~2 (or at least~6 for graphs with maximum degree at least~6)  is incident with at least two edges of different weights, then $\chi'^{\Delta-1}_{\sum}(G)\le k$.
	
	Our paper is organised as follows. 
	In Sections~\ref{sec:degreeFour} and~\ref{sec:degreeFive}, we consider nice graphs with degree at most~4 and at most~5, respectively.
	We prove that every nice graph $G$ with degree at most~5 admits a neighbour sum distinguishing $(\Delta(G)+2)$-edge-weighting such that all the vertices of degree at least~2 are incident with at least two edges of different weights. In Section~\ref{sec:general},  we prove that every nice graph  admits a neighbour sum distinguishing $7$-edge-weighting such that all the vertices of degree at least~6 are incident with at least two edges of different weights. In Section~\ref{sec:bipartite}, we show that the result from Section~\ref{sec:general} can be improved for bipartite graphs:  we prove that every  nice bipartite graph   admits  a neighbour sum distinguishing $6$-edge-weighting such that all the vertices of degree at least~2 are incident with at least two edges of different weights. Furthermore, we show that every connected bipartite graph on at least three vertices having a vertex partition $(V_1,V_2)$ such that $|V_1|$ is even admits a neighbour sum  distinguishing $4$-edge-weighting such that every vertex of degree at least~2 is incident with at least two edges of different weights.
	
	Those results can be reframed in the relaxed framework. First, we prove in Sections~\ref{sec:degreeFour} and~\ref{sec:degreeFive} every nice graph $G$ with $\Delta(G) \leq 4$ (resp. $\Delta(G) \leq 5$) verifies $\chi'^{\Delta-1}_{\sum}(G)\le 6$ (resp. $\chi'^{\Delta-1}_{\sum}(G)\le 7$). Then, we prove in Section~\ref{sec:general} a result that, together with the previous two, shows that every nice graph $G$ verifies $\chi'^{\Delta-1}_{\sum}(G)\le 7$. Finally, we prove in Section~\ref{sec:bipartite} that every nice bipartite graph $G$ verifies $\chi'^{\Delta-1}_{\sum}(G)\le 6$. Note that, for all those graphs, the expected bound from the general conjecture is~4. While the bounds we prove are quite close, there is still room for improvement.
	
	\ \\
	
	We will need several preliminary results from the literature in order to prove our own, the first of which being the theorem in~\cite{DaDuPeSi} about the neighbour sum distinguishing $2$-relaxed $4$-edge-colouring of subcubic graphs, which can be equivalently rewritten in the following way:
	
	\begin{theorem}{\rm \cite{DaDuPeSi}}\label{thm:subcubic}
		If $G$ is a nice subcubic graph with no  component isomorphic to $C_5$, then it admits a neighbour sum  distinguishing $4$-edge-weighting such that every vertex of degree at least 2 is incident with at least two edges of different weights.
	\end{theorem}

	Another result that will be widely used for graphs with maximum degree at most~5 is the following theorem by Alon \cite{Al99}.
	
	\begin{theorem}[Combinatorial Nullstellensatz \cite{Al99}] \label{alon}
		Let $\mathbb{F}$ be an arbitrary field, and let $P = P(x_1,\ldots, x_n)$ be a polynomial in $\mathbb{F}[x_1,\ldots, x_n]$. Suppose the degree $deg(P)$ of $P$ equals $\sum_{i=1}^n k_i$, where each $k_i$ is a nonnegative integer, and suppose the coefficient of $x_1^{k_1}\cdots x_n^{k_n}$ in $P$ is nonzero. Then if $S_1,\ldots , S_n$ are subsets of $\mathbb{F}$ with $|S_i| > k_i$, there are $s_1\in S_1,\ldots ,s_n\in S_n$ such that $P(s_1,\ldots, s_n)\neq 0$.
	\end{theorem}

	
	\section{Graphs with maximum degree at most 4}
	\label{sec:degreeFour}

	\begin{theorem}\label{thm:four}
		Every   nice  graph  $G$  with $\Delta(G)\le 4$ admits a neighbour sum distinguishing $6$-edge-weighting  such that all the vertices of degree at least $2$ are incident with at least two edges of different weights.
	\end{theorem}

	This result can be restated the following way:
	
	\begin{reptheorem}{thm:four}
		Every nice  graph  $G$  with $\Delta(G)\le 4$ verifies $\chi'^{\Delta-1}_{\sum}(G)\le 6$.
	\end{reptheorem}

	\begin{proof}
		We proceed by induction on the number of edges. It is easy to see that the theorem is true for nice graphs with two  and three edges. Assume that the theorem is true for nice graphs with at most $m-1$ edges.   Let $G$ be a  nice graph with $m$ edges. We may assume that $G$ is connected, since otherwise, by induction, every component has a $6$-edge-weighting that satisfies the theorem. Furthermore, we  may assume that  $\Delta(G)=4$, because, by Theorem \ref{thm:subcubic}, the theorem is true for all  nice subcubic graphs except $C_5$, and $C_5$ admits  a neighbour sum distinguishing $5$-edge-weighting such that all the vertices  are incident with edges of different weights. Let $u$ be a vertex of degree 4 in $G$.

		\medskip
		\noindent {\bf Case 1. There is an edge in the subgraph induced by $N(u)$}
		
		Let $N(u)=\{v,w,u_1,u_2\}$ and $vw\in E(G)$. Let $G'$ be obtained from $G$ by removing the two edges $uv$ and $uw$. $G'$ has at most two components. Each component of  $G'$ with at least two edges admits an edge-weighting  that satisfies the theorem. If $d_{G'}(v)\ge 2$ or $d_{G'}(w)\ge 2$, then every component has at least two edges, otherwise one component is isomorphic to $K_2$. Let $\om$ be an edge-weighting of components of $G'$ with at least two edges that satisfies the theorem, and   additionally we extend the edge-weighting $\om$ on the component isomorphic to $K_2$ (if such exists), which we weight with an arbitrary weight.

		To obtain our final  edge-weighting, we just need to weight the two edges $uv$ and  $uw$ while making sure that all the vertices of $\{u,v,w\}$ are distinguished with their neighbours and the vertices $v$ and $w$ are incident with two edges of different weights. Note that the
		vertex $u$ already has two incident edges of distinct weights, because $d_{G'}(u)=2$. If $d_{G'}(v)\ge 2$ and $d_{G'}(w)\ge 2$, then  $v$ and $w$ also have two incident edges of distinct weights, otherwise we have to choose a weight on $uv$ and  $uw$ that is different from $\om(vw)$.

		First, we consider how many weights we have to forbid for the edges $uv$ and $uw$ such that we obtain an edge-weighting that distinguishes all adjacent vertices except the pairs $(v,w),(u,u_1),(u,u_2)$ and such that all vertices of degree at least 2 are incident with two edges of different weights. The vertex $v$ must be distinguished from its neighbours in $G'-w$. If $v$ has two neighbours in $G'-w$, then  there are potentially two forbidden weights for $vu$. Thus, four possible weights remain for $uv$. If $v$ has exactly one neighbour in $G'-w$, then there is at most one forbidden weight for $uv$. If $w$ is the only neighbour of $v$ in $G'$, then the weight of $uv$ must be different from the weight of $vw$ and hence there is at most one forbidden weight for $uv$. Thus, summarising, there are at most two forbidden weights for $uv$.  Similarly, we  can observe that for $uw$ there are at most two forbidden weights. Let $S_1$ be the set of weights that are not forbidden for $uv$ and  $S_2$ be the set of weights that are not forbidden for $uw$, so $|S_1|\ge 4$ and $|S_2|\ge 4$. Observe that if we choose for $uv$ a weight from $S_1$ and for $uw$ a weight from $S_2$, then we obtain  an edge-weighting of $G$ in which all pairs of adjacent vertices, except $(v,w),(u,u_1),(u,u_2)$, are distinguished, and all vertices of degree at least 2 are incident with two edges of different weights. Let $x_1\in S_1$ and $x_2\in S_2$ be the weights attributed to $uv$ and $uw$, respectively.  To obtain an edge-weighting that satisfies the theorem for the weights $x_1$ and $x_2$, we must have the following: 
		\vspace{-5pt}
		\begin{itemize}\itemsep-3pt
			\item $x_1+x_2+\sigma_{\om}(u)\neq \sigma_{\om}(u_i)$, because $u$ must be distinguished from $u_i$ for $i\in\{1,2\}$;
			\item $x_2+\sigma_{\om}(u)\neq \sigma_{\om}(v)$, because we have to distinguish $u$ and $v$;
			\item $x_1+\sigma_{\om}(u)\neq \sigma_{\om}(w)$, because we have to distinguish $u$ and $w$;
			\item $x_1+\sigma_{\om}(v)\neq x_2+\sigma_{\om}(w)$, because we have to distinguish $v$ and $w$.
		\end{itemize}
		\noindent To prove that there are weights $x_1$ and $x_2$ that satisfy all the above conditions, we construct the polynomial: 
		
		\begin{align*}
			P(x_1,x_2) ={} & (x_1+x_2+\sigma_{\om}(u)- \sigma_{\om}(u_1)) \\
			& (x_1+x_2+\sigma_{\om}(u)- \sigma_{\om}(u_2)) \\
			& (x_2+\sigma_{\om}(u)- \sigma_{\om}(v)) \\
			& (x_1+\sigma_{\om}(u)- \sigma_{\om}(w)) \\
			& (x_1-x_2+\sigma_{\om}(v)-\sigma_{\om}(w)).
		\end{align*}

		\noindent If there exist $x_1$ and $x_2$ such that $P(x_1,x_2)\neq 0$ and  $x_i\in S_i$ ($i\in\{1,2\}$), then the $x_i$'s satisfy all the conditions. By weighting $uv,uw$ with $x_1,x_2$,  we can extend the edge-weighting $\om$ to an edge-weighting that satisfies the theorem.  We apply Theorem  \ref{alon} to prove that $x_1$ and $x_2$ exist. First, we claim that the coefficient of the monomial  $x_1^3x_2^2$ is non-zero. Observe that this coefficient in $P$ is the same as in the following polynomial:
		
		$$P_1(x_1,x_2)=(x_1+x_2)^2(x_1-x_2)x_1x_2.$$

		\noindent The coefficient of the monomial  $x_1^3x_2^2$ is 1. Since $|S_1|>3$ and $|S_2|>2$,  Theorem \ref{alon} implies that there are $x_1\in S_1$ and $x_2\in S_2$ such that $P(x_1,x_2)\neq 0$ and equivalently there is the desired  edge-weighting of $G$. 
		
		\medskip
		\noindent {\bf Case 2.  $N(u)$ is an independent set}
		
		Let $N(u)=\{u_1,u_2, u_3,u_4\}$ and $G'=G-u$. Each component of  $G'$ with at least two edges admits an edge-weighting  that satisfies the theorem. Let $\om$ be an edge-weighting of components of $G'$ with at least two edges that satisfies the theorem, and additionally we extend the edge-weighting $\om$ to the components isomorphic to $K_2$ (if such exists), which we weight with an arbitrary weight.  
		
		To obtain a final  edge-weighting, we just need to weight the  edges $uu_i$ for $i\in\{1,2,3,4\}$. We choose a weight for $uu_i$ in such a way that ensures that $u_i$ is distinguished with its neighbours in $G'$ and if $u_i$ has exactly one neighbour in $G'$, then the weight of $uu_i$ is different from the weight of the edge incident with $u_i$ in $G'$. Furthermore, after weighting the four edges $uu_1,uu_2,uu_3,uu_4$, the vertex $u$ must be distinguished from its neighbours, and these edges cannot be weighted with the same weight.
		
		First, we consider how many weights we have to forbid for edges $uu_i$ such that we obtain an edge-weighting in which the pairs of adjacent vertices of $G'$ are still distinguished and all vertices of $G'$ are incident with two edges of distinct weights. Since  $u_i$ must be distinguished from its neighbours in $G'$, we have at most three forbidden weights for $uu_i$. If $u_i$ has exactly one neighbour in $G'$, then in order to distinguish $u_i$ from its neighbour there is at most one forbidden weight and the weight of $uu_i$ must be different from the weight of the edge incident with $u_i$ in $G'$, so together we have at most two forbidden weights. Thus, $uu_i$ has at most three forbidden weights.
		Let $S_i$ be the set of weights that are not forbidden for $uu_i$, thus $|S_i|\ge 3$ for $i\in\{1,2,3,4\}$. After weighting the edge $uu_i$ with weight $x_i\in S_i$ for $i\in\{1,2,3,4\}$ we obtain an edge-weighting that distinguishes all vertices of $G'$ and every vertex of $G'$ is incident with at least two edges of different weights. Let $x_i\in S_i$ be weights attributed to $uu_i$ for $i\in\{1,2,3,4\}$.  To obtain an edge-weighting that satisfies the theorem for  $x_i$, it must additionally  hold: 
		\vspace{-5pt}
		\begin{itemize}\itemsep-3pt
			\item $x_1+x_2+x_3+x_4-x_i\neq \sigma_{\om}(u_i)$, because we have to distinguish $u$ and $u_i$ for $i\in\{1,2,3,4\}$;
			\item $x_i\neq x_j$ for some $i,j\in \{1,2,3,4\}$, because $u$ must be  incident with at least two edges of different weights.
		\end{itemize}
		\noindent We consider the polynomial 
		
		\begin{align*}
			P(x_1,x_2,x_3,x_4)= & (x_2+x_3+x_4-\sigma_{\om}(u_1)) \\
			& (x_1+x_3+x_4- \sigma_{\om}(u_2)) \\
			& (x_1+x_2+x_4- \sigma_{\om}(u_3)) \\
			& (x_1+x_2+x_3- \sigma_{\om}(u_4)) \\
			& (x_3-x_4).
		\end{align*}
		
		\noindent If there exist $x_1,x_2,x_3,x_4$ such that $P(x_1,x_2,x_3,x_4)\neq 0$ and $x_i\in S_i$ ($i\in\{1,2,3,4\}$), then the $x_i$ satisfy all the conditions. By weighting $uu_i$ with $x_i$, we can extend the edge-weighting $\om$ to an edge-weighting that satisfies the theorem. To prove that there are such $x_i$ we again apply Theorem  \ref{alon}. We consider the coefficient of the monomial  $x_1^2x_2x_3^2$. Observe that this coefficient in $P$ is the same as in the following polynomial:		
		$$P_1(x_1,x_2,x_3,x_4)=(x_2+x_3+x_4)(x_1+x_3+x_4)(x_1+x_2+x_4)(x_1+x_2+x_3)(x_3-x_4).$$
		
		\noindent The coefficient of the monomial  $x_1^2x_2x_3^2$ is non-zero. Since $|S_1|>2,|S_2|>1$ and $|S_3|\ge 2$,  Theorem \ref{alon} implies that there are $x_i\in S_i$ ($i\in\{1,2,3,4\}$) such that $P(x_1,x_2,x_3,x_4)\neq 0$ and so an edge-weighting of $G$ that satisfies the theorem exists. 
		
	\end{proof}
	
	
	\section{Graphs with maximum degree at most 5}
	\label{sec:degreeFive}

	\begin{theorem}\label{thm:five}
		Every   nice  graph  $G$  with $\Delta(G)\le 5$ admits a neighbour sum distinguishing $7$-edge-weighting  such that all the vertices of degree at least $2$ are incident with at least two edges of different weights.
	\end{theorem}

	This result can be restated the following way:
	
	\begin{reptheorem}{thm:five}
		Every nice  graph  $G$  with $\Delta(G)\le 5$ verifies $\chi'^{\Delta-1}_{\sum}(G)\le 7$.
	\end{reptheorem}
	
	\begin{proof}
		We proceed by induction on the number of edges.  It is easy to see that the theorem is true for nice graphs with two,  three and four edges.  Assume that the theorem is true for nice graphs with at most $m-1$ edges. Let $G$ be a  nice graph with $m$ edges. We may assume that $G$ is connected, since otherwise, by induction, every component admits a $7$-edge-weighting that satisfies the theorem. Furthermore, by Theorem \ref{thm:four}, we  may assume that  $\Delta(G)=5$ since, otherwise, the result holds.  Let $u$ be a vertex of degree 5 in $G$. 
		
		\medskip
		\noindent {\bf Case 1. There is an edge in the subgraph induced by $N(u)$}
		
		Let $N(u)=\{v,w,u_1,u_2,u_3\}$ and $vw\in E(G)$. 
		
		First, we consider the case where $d_G(v)\le 3$ or $d_G(w)\le 3$, say without loss of generality $d_G(v)\le 3$.
		Let $G'$ be the graph obtained by removing from $G$ the two edges $uv$ and $uw$. Each component of  $G'$ with at least two edges admits an edge-weighting  that satisfies the theorem.  Let $\om$ be an  edge-weighting of components of $G'$ with at least two edges that satisfies the theorem, and   additionally we extend the edge-weighting $\om$ on the components isomorphic to $K_2$, which we weight with an arbitrary weight.

		To obtain our desired edge-weighting, we need to weight the two edges $uv$ and  $uw$,  making sure that the vertices $u$, $v$ and $w$ are distinguished from their neighbours and the vertices $v$ and $w$ are incident with two edges of distinct weights. Note that the
		vertex $u$ already verifies this property, since $d_{G'}(u)=3$.

		First, we consider how many weights we have to forbid for the edges $uv$ and $uw$ for us to obtain an edge-weighting that distinguishes all adjacent vertices except the pairs $(v,w),(u,u_1), (u,u_2), (u,u_3)$ and in which all vertices of degree at least two are incident with two edges of different weights. The vertex $v$ must be distinguished from its neighbour in $G'-w$. If $v$ has one neighbour in $G'-w$, then  there is potentially one forbidden weight for $vu$, such that $v$ is just incident with two edges weighted differently. If $v$ has no neighbour in $G'-w$, then the weight of $uv$ must be different from the weight of $vw$; so again there is  one forbidden weight for $uv$. Thus, there are six possible weights for $uv$.  To distinguish $w$ from its neighbours in $G'-w$ there are at most three forbidden weights.  If $w$ is the only neighbour of $v$ in $G'$, then the weight of $uw$ must be different from the weight of $vw$ and hence there is at most one forbidden weight for $uv$.   In conclusion, there are at least four possible weights for $uw$.  Let $S_1$ be the set of weights that are not forbidden for $uv$ and  $S_2$ be the set of weights that are not forbidden for $uw$, so $|S_1|\ge 6$ and $|S_2|\ge 4$. To prove that we can choose weights from $S_1$ and $S_2$ such that we result in an edge-weighting that satisfies the conditions of the theorem, we use Theorem  \ref{alon}.  Let $x_1\in S_1$ and $x_2\in S_2$ be weights attributed to $uv$ and $uw$, respectively.  To obtain the final edge-weighting, the weights $x_1$ and $x_2$ must additionally verify: 
		\vspace{-5pt}
		\begin{itemize}\itemsep-3pt
			\item $x_1+x_2+\sigma_{\om}(u)\neq \sigma_{\om}(u_i)$, because $u$ must be distinguished from $u_i$, for $i=1,2,3$;
			\item $x_2+\sigma_{\om}(u)\neq \sigma_{\om}(v)$, because we have to distinguish $u$ and $v$;
			\item $x_1+\sigma_{\om}(u)\neq \sigma_{\om}(w)$, because we have to distinguish $u$ and $w$;
			\item $x_1+\sigma_{\om}(v)\neq x_2+\sigma_{\om}(w)$, because we have to distinguish $v$ and $w$.
		\end{itemize}
		\noindent We construct the polynomial 
		
		\begin{align*}
			P(x_1,x_2)= & \prod_{i=1,2,3}(x_1+x_2+\sigma_{\om}(u)- \sigma_{\om}(u_i)) \\
			& (x_2+\sigma_{\om}(u)- \sigma_{\om}(v)) \\
			& (x_1+\sigma_{\om}(u)- \sigma_{\om}(w)) \\
			& (x_1-x_2+\sigma_{\om}(v)-\sigma_{\om}(w)).
		\end{align*}

		\noindent We consider the coefficient of the monomial  $x_1^5x_2$. Observe that this coefficient in $P$ is the same as in the following polynomial:
		$$ P_1(x_1,x_2)=(x_1+x_2)^3(x_1-x_2)x_1x_2.$$

		\noindent The coefficient of the monomial  $x_1^5x_2$ is 1. Since $|S_1|>5$ and $|S_2|>1$,  Theorem \ref{alon} implies that there are $x_1\in S_1$ and $x_2\in S_2$ such that $P(x_1,x_2)\neq 0$ and equivalently we can construct the desired  edge-weighting of $G$. 
		
		Consider now the case when $d_G(v)\ge 4$ and $d_G(w)\ge 4$.
		Let $G'$ be obtained from $G$ by removing the three edges $uv,uw$ and $vw$. Each component of  $G'$ has at least two edges, so it admits an edge-weighting  that satisfies the theorem.  Let $\om$ be an edge-weighting of components of $G'$  that satisfies the theorem. Observe that in $G'$ the vertices $u,v,w$ are  just incident with at least two edges of different weights, since $d_{G'}(u)=3,d_{G'}(v)\ge 2$ and $d_{G'}(w)\ge 2$. 
		
		Let $x_1,x_2$ and $x_3$ be weights attributed to $uv,uw$ and $vw$, respectively.  To obtain an edge-weighting of $G$ that satisfies the theorem, $x_1$ and $x_2$ must verify: 
		\vspace{-5pt}
		\begin{itemize}\itemsep-3pt
			\item $x_1+x_2+\sigma_{\om}(u)\neq \sigma_{\om}(u_i)$, because $u$ must be distinguished from $u_i$ for $i\in\{1,2,3\}$;
			\item $x_1+x_3+\sigma_{\om}(v)\neq \sigma_{\om}(v_i)$, where $i\in \{1,2\}$ if $v$ has two neighbours $v_1,v_2$ in $G'$ and $i\in \{1,2,3\}$ if $v$ has three neighbours $v_1,v_2,v_3$ in $G'$, because $v$ must be distinguished from its neighbours in $G'$;
			\item $x_2+x_3+\sigma_{\om}(w)\neq \sigma_{\om}(w_i)$, where $i\in \{1,2\}$ if $w$ has two neighbours $w_1,w_2$ in $G'$ and $i\in \{1,2,3\}$ if $w$ has three neighbours $w_1,w_2,w_3$ in $G'$, because $w$ must be distinguished from its neighbours in $G'$;
			\item $x_1+\sigma_{\om}(v)\neq x_2+\sigma_{\om}(w)$, because $v$ must be distinguished from $w$;
			\item $x_1+\sigma_{\om}(u)\neq x_3+\sigma_{\om}(w)$, because $u$ must be distinguished from $w$;
			\item $x_2+\sigma_{\om}(u)\neq x_3+\sigma_{\om}(v)$, because $u$ must be distinguished from $v$.
		\end{itemize}
		
		\noindent We construct the polynomial 
		
		\begin{align*}
			P(x_1,x_2,x_3)= & \prod_{i=1,2,3}(x_1+x_2+\sigma_{\om}(u)- \sigma_{\om}(u_i)) \\
			& \prod_{i=1,2,3}(x_1+x_3+\sigma_{\om}(v)- \sigma_{\om}(v_i)) \\
			& \prod_{i=1,2,3}(x_2+x_3+\sigma_{\om}(w)- \sigma_{\om}(w_i)) \\
			& (x_1+\sigma_{\om}(v)- x_2-\sigma_{\om}(w)) \\
			& (x_1+\sigma_{\om}(u)- x_3-\sigma_{\om}(w)) \\
			& (x_2+\sigma_{\om}(u)- x_3-\sigma_{\om}(v)).
		\end{align*}
		
		\noindent If there are $x_i\in \{1,\ldots,7\}$ ($i\in\{1,2,3\}$) such that $P(x_1,x_2,x_3)\neq 0$, then by weighting $uv,uw,vw$ with $x_1,x_2,x_3$ we can extend the edge-weighting $\om$ of $G'$ to an edge-weighting of $G$ that satisfies the theorem whenever $d_G(v)=d_G(w)=5$.  If $d_G(v)= 4$ or $d_G(w)= 4$, then the polynomial $R$,  which we should construct for proving that the weights $x_1,x_2,x_3$ exist, is a factor of $P(x_1,x_2,x_3)$. However, if $P(x_1,x_2,x_3)\neq 0$, then also for the factor $R$ we have $R(x_1,x_2,x_3)\neq 0$. So it is enough to consider  the polynomial $P$.
		
		To prove that there are $x_i\in \{1,\ldots,7\}$ ($i\in\{1,2,3\}$) such that $P(x_1,x_2,x_3)\neq 0$ we apply Theorem  \ref{alon}. Consider the coefficient of the monomial  $x_1^5x_2^4x_3^3$. Observe that this coefficient in $P$ is the same as in the following polynomial:
		$$P_1(x_1,x_2,x_3)=(x_1+x_2)^3(x_1+x_3)^3(x_2+x_3)^3(x_1-x_2)(x_1-x_3)(x_2-x_3).$$
		
		\noindent The coefficient of the monomial  $x_1^5x_2^4x_3^3$ is 2. Theorem \ref{alon} implies that there are $x_i\in \{1,\ldots,7\}$  such that $P(x_1,x_2,x_3)\neq 0$ and equivalently there is the desired  edge-weighting of $G$.

		\medskip
		\noindent {\bf Case 2.  $N(u)$ is an independent set}
		
		This part of the proof is very similar to {\bf Case 2} of the proof of Theorem \ref{thm:four}.
		Let $N(u)=\{u_1,u_2, u_3,u_4,u_5\}$. Let $G'=G-u$. Each component of  $G'$ with at least two edges has an edge-weighting  that satisfies the theorem. Let $\om$ be an  edge-weighting of components of $G'$ with at least two edges that satisfies the theorem. We extend the edge-weighting $\om$ to the components isomorphic to $K_2$, which we weight with an arbitrary weight.  
		
		First, we consider how many weights we have to forbid for edges $uu_i$ such that we result in an edge-weighting in which the pairs of adjacent vertices of $G'$ are still distinguished and all vertices of $G'$ are incident with two edges of distinct weights. Since the vertex $u_i$ must be distinguished from its neighbours in $G'$ and $d_{G'}(u_i) \leq 4$, we have at most four forbidden weights for $uu_i$. If $u_i$ has exactly one neighbour in $G'$, then in order to distinguish $u_i$ from its neighbour there is at most one forbidden weight and the weight of $uu_i$ must be different from the weight of the edge incident with $u_i$ in $G'$, so together we have at most two forbidden weights. Let $S_i$ be the set of weights that are not forbidden for $uu_i$, thus $|S_i|\ge 3$ for $i\in\{1,2,3,4,5\}$. After weighting the edge $uu_i$ with weight $x_i\in S_i$ for $i\in\{1,2,3,4,5\}$, we obtain an edge-weighting that distinguishes all vertices of $G'$ and every vertex of $G'$ is incident with at least two edges of different weights. Let $x_i\in S_i$ be weights attributed to $uu_i$ for $i\in\{1,2,3,4,5\}$.  To obtain an edge-weighting that satisfies the theorem, the weights $x_i$ must additionally  verify: 
		\vspace{-5pt}
		\begin{itemize}\itemsep-3pt
			\item $x_1+x_2+x_3+x_4+x_5-x_i\neq \sigma_{\om}(u_i)$, because we have to distinguish $u$ and $u_i$ for $i\in\{1,2,3,4,5\}$;
			\item $x_i\neq x_j$ for some $i,j\in \{1,2,3,4,5\}$, because $u$ must be  adjacent to at least two edges of different weights.
		\end{itemize}
		
		\noindent We construct the polynomial 
		
		\begin{align*}
			P(x_,x_2,x_3,x_4,x_5)= & (x_2+x_3+x_4+x_5-\sigma_{\om}(u_1)) \\
			& (x_1+x_3+x_4+x_5- \sigma_{\om}(u_2)) \\
			& (x_1+x_2+x_4+x_5- \sigma_{\om}(u_3)) \\
			& (x_1+x_2+x_3+x_5- \sigma_{\om}(u_4)) \\
			& (x_1+x_2+x_3+x_4-\sigma_{\om}(u_5)) \\
			& (x_3-x_4).
		\end{align*}
		
		\noindent If there are $x_i\in S_i$ ($i\in\{1,2,3,4,5\}$) such that $P(x_1,x_2,x_3,x_4,x_5)\neq 0$, then, by weighting $uu_i$ with $x_i$, we extend the edge-weighting $\om$ to an edge-weighting that satisfies the theorem. We again apply Theorem  \ref{alon} to prove that there are such $x_i$'s. We consider the coefficient of the monomial  $x_1^2x_2^2x_3^2$. Observe that this coefficient in $P$ is the same as in the following polynomial:
		
		\begin{align*}
                  P_1(x_1,x_2,x_3,x_4,x_5)=&(x_2+x_3+x_4+x_5)(x_1+x_3+x_4+x_5)(x_1+x_2+x_4+x_5)\\
                  &(x_1+x_2+x_3+x_5)(x_1+x_2+x_3+x_4)(x_3-x_4).
                  \end{align*}
		
		\noindent The coefficient of the monomial  $x_1^2x_2^2x_3^2$ is non-zero. Since $|S_1|>2,|S_2|>2$ and $|S_3|> 2$,  Theorem \ref{alon} implies that there are $x_i\in S_i$ ($i\in\{1,2,3,4,5\}$) such that $P(x_1,x_2,x_3,x_4,x_5)\neq 0$ and so we can construct the desired  edge-weighting of $G$. 
	\end{proof}

	
	\section{Graphs with maximum degree at least 6}
	\label{sec:general}

	In this section, we prove that every nice graph admits a neighbour sum distinguishing $7$-edge-weighting such that every vertex of degree at least 6 is incident with at least two edges of different weights. Our approach is based on the algorithm given in   \cite{KaLu04} for proving that every nice graph admits a neighbour sum distinguishing $5$-edge-weighting. It is worth mentioning that modifications of that algorithm allowed getting new results for the neighbour sum distinguishing edge-weighting and its variants. For example, Bensmail \cite{Be19} proved that every 5-regular graph admits a neighbour sum distinguishing $4$-edge-weighting and Gao et al. \cite{GaWaWu16} proved that the 1-2-3 Conjecture is true if we allow the vertices with the same incident sum to induce a forest.
	
	We prove the following theorem:
	\begin{theorem}
		\label{thm:seven_colours}
		Every   nice  graph  $G$   admits a neighbour sum distinguishing $7$-edge-weighting of $G$ such that all the vertices of degree at least $6$ are incident with at least two edges of different weights.
	\end{theorem}
	
	This, together with Theorem~10 in~\cite{DaDuPeSi} (for subcubic graphs) and Theorems~\ref{thm:four} and~\ref{thm:five} (for maximum degrees~4 and~5), allows us to have the following general result:
	
	\begin{corollary}
		Every nice  graph $G$ verifies $\chi'^{\Delta-1}_{\sum}(G)\le 7$.
	\end{corollary}
	
	\noindent {\it Rough ideas of the proof of Theorem \ref{thm:seven_colours}}
	
	
	We give an algorithm which constructs a vertex-colouring $w$ and a $7$-edge-weighting $\om$.  The vertex-colouring $w$ will be almost the vertex-colouring $\sigma_{\om}$, namely $\sigma_{\om}(u)=w(u)$ or $\sigma_{\om}(u)=w(u)+3$ for $u\in V(G)$. The $7$-edge-weighting $\om$ will satisfy the conditions of Theorem \ref{thm:seven_colours}. The algorithm processes the vertices one after another, following a special ordering. First, we define that ordering and prove that every nice graph, except stars, admits such an ordering of vertices. Then, we give the algorithm and prove that every step of the algorithm is always executable. Finally, we prove that the $7$-edge-weighting $\om$ given by the algorithm is neighbour sum distinguishing and that all  vertices of degree at least $6$ are incident with at least two edges of different weights.
\vspace{0.5cm}

	Before we define the ordering of vertices (in Lemma \ref{lem:vertex_order}) we need the following notations.
	
	Let $(v_1,v_2,\ldots ,v_n)$ be an ordering of vertices of $G$. We say that $v_j$ {\it follows} $v_i$ in the ordering if $i<j$. A \emph{predecessor} (resp. \emph{successor}) of $v_i$ is every neighbour of $v_i$ in $\{v_1,\ldots ,v_{i-1}\}$ (resp. in $\{v_{i+1},\ldots ,v_n\}$) for $i\in \{1,\ldots,n\}$. Let us define a partial ordering induced by a given vertex ordering $(v_1,v_2,\ldots ,v_n)$ in the following way
	$$v_j\prec v_i \Leftrightarrow \;\mbox{there is a path}\; v_j v_{k_1} v_{k_2} \ldots v_{k_d} v_i\; \mbox{in} \;G\; \mbox{such that}\; j<k_1<k_2<\ldots<k_d<i.$$

	\begin{remark}
		Observe that two different vertex orderings of the graph $G$ may induce the same partial ordering. Indeed, let $(v_1,v_2,\ldots,v_i,v_{i+1},\ldots  ,v_n)$ be an ordering of vertices of $G$ such that $v_iv_{i+1}\notin E(G)$. Thus, $v_i\not\prec v_{i+1}$. The ordering $(v_1,v_2,\ldots,v_{i+1},v_{i},\ldots  ,v_n)$ induces the same partial ordering as $(v_1,v_2,\ldots,v_i,v_{i+1},\ldots  ,v_n)$.
	\end{remark}
	
	\begin{remark}
		If  $y\prec x$, then $x$ has a predecessor and $y$ has a successor.
	\end{remark}
	The {\it inversion} of the ordering $(v_1,v_2,\ldots ,v_n)$ is the ordering $(v_n,v_{n-1},\ldots ,v_1)$.
	
	\begin{lemma}\label{lem:vertex_order}
		Let $G$ be  a connected graph  on $n$ vertices and $G\neq K_{1,n-1}$. There is a vertex ordering $(v_1,v_2,\ldots ,v_n)$ of $G$ such that
		\begin{enumerate}[(i)]
			\item $d(v_1)\ge 2$ and $d(v_2)\ge 2$;
			\item $v_i$ has a predecessor for $i\in \{2,\ldots, n\}$;
			\item if $v_i$ has no successor, then, in $N_G(v_i)$, there is at most one vertex having a successor in $\{v_{i+1},\ldots, v_n\}$ for $i\in \{1,\ldots,n\}$.
		\end{enumerate}
	\end{lemma}
	
	\begin{remark}
		The condition (ii) can be equivalently replaced by the following one: $v_1\prec v_i$ for $i\in \{2,\ldots, n\}$.
	\end{remark}
	
	\begin{proof} [of Lemma \ref{lem:vertex_order}]
		It is easy to see that if $G$ is a connected graph  and $G$ is not a star, then  there is an ordering that satisfies  conditions (i) and (ii). On the contrary, suppose that there is no ordering that satisfies (i), (ii), and (iii).  For an ordering ${\bf v}=(v_1,v_2,\ldots ,v_n)$ by $B({\bf v})$ we denote the set  of vertices   which have no successor.
		
		Let ${\bf v}=(v_1,v_2,\ldots ,v_n)$ be an ordering that satisfies (i) and (ii) minimises $|B({\bf v})|$.

		Let $\prec$ be the partial ordering induced by ${\bf v}$ and $v$ be the first vertex in ${\bf v}$ for which (iii) fails, so $v\in B({\bf v})$.  Let ${\bf v'}$ be an ordering of $V(G)$ which induces the same partial ordering $\prec$ as ${\bf v}$, but in which the index of $v$ is minimum and let $v=v_i$ in ${\bf v'}$. Observe that every vertex has the same predecessors and successors in both orderings; so $|B({\bf v'})|=|B({\bf v})|$ and the vertex $v$ still makes (iii) fail in the ordering ${\bf v'}$. Furthermore, the choice of ${\bf v'}$ implies that for any $x\in \{v_1,\ldots, v_{i-1}\}$ we have $x\prec v_i$. 
		Let $j$ be the largest integer smaller than $i$ such that $v_j$ is a predecessor of $v_i$ and $v_j$ has a successor in $\{v_{i+1},\ldots, v_n\}$. 
		
		{\bf Case 1. $j>3$}
		
		Let ${\bf w}=(v_j,v_{k_1},v_{k_2},\ldots, v_{k_\ell})$ be a subordering of ${\bf v'}$ containing $v_j$ and all vertices $x$ such that $v_j\prec x \prec v_i$. Let ${\bf w'}$ be the inverse of ${\bf w}$.  We reorder the vertices of ${\bf v'}$ in the following way: ${\bf v''}=(v_1,\ldots,v_{j-1}, v_i,{\bf w'},v_{i+1},\ldots,v_{n})$. Let $\prec '$ be the partial ordering induced by ${\bf v''}$.  Since $v_j$ was the last predecessor of $v_i$ having  a successor in $\{v_{i+1},\ldots, v_n\}$, $v_i$  still has predecessor in ${\bf v''}$ and now $v_i$ has a successor. Furthermore, for any $x\in {\bf w'}\setminus\{v_j\}$ we have $v_i\prec' x\prec ' v_j$ and hence every vertex of ${\bf w'}\setminus\{v_j\}$ has a predecessor and a successor. Also $v_j$ has a predecessor and a successor in ${\bf v''}$. Thus ${\bf v''}$ satisfies  conditions (i) and (ii) and  $|B({\bf v''})|<|B({\bf v})|$, a contradiction.
		
		{\bf Case 2. $j\le 2$}
		
		Since $v_i$ has at least two predecessors having  a successor in $\{v_{i+1},\ldots, v_n\}$, $j=2$ and $v_1, v_2$ have   successors in $\{v_{i+1},\ldots, v_n\}$. Furthermore, $v_1, v_2$ are the only predecessor of $v_i$ having   successors in $\{v_{i+1},\ldots, v_n\}$. If $i=3$, then we reorder the vertices of ${\bf v'}$ in the following way: ${\bf v''}=(v_1,v_3,v_2,v_4,\ldots,v_n)$. In ${\bf v''}$ the vertex $v_3$ has a successor; so  $|B({\bf v''})|<|B({\bf v})|$ and ${\bf v''}$ satisfies  conditions (i) and (ii), a contradiction. Suppose that $i>3$. The condition (ii) implies that $v_3$ is adjacent to $v_2$ or $v_1$. If $v_3v_2\in E(G)$, then we reorder  ${\bf v'}$ in the following way: ${\bf v''}=(v_2,v_3,v_4,\ldots, v_i,v_1,v_{i+1}\ldots,v_n)$. If $v_2v_2\notin E(G)$, then we reorder  ${\bf v'}$ in the following way: ${\bf v''}=(v_1,v_3,v_4,\ldots, v_i,v_2,v_{i+1},\ldots,v_n)$. In both cases, ${\bf v''}$ satisfies  conditions (i) and (ii) and  $|B({\bf v''})|<|B({\bf v})|$, a contradiction. 
	\end{proof}
	
	\vspace{0.5cm}
	\noindent {\bf ALGORITHM}
	
	Let $G$ be an $n$-vertex connected graph  and $G\neq K_{1,n-1}$. Let ${\bf v}=(v_1,v_2,\ldots ,v_n)$ be a vertex ordering that satisfies  conditions (i)--(iii) of Lemma \ref{lem:vertex_order}. Let 
	
	\noindent $V'=\{v_i\in \{v_1,v_2,\ldots ,v_n\}: v_i\; \mbox{has a successor}\}$, 
	
	\noindent $V''=\{v_i\in \{v_1,v_2,\ldots ,v_n\}: v_i \;\mbox{has no successor}\}$.
	
	We start by assigning the provisional weight $4$ to every edge, then we process the $v_i's$ one after another, following the ordering ${\bf v}$. Whenever we treat  a new vertex $v_i$, we modify the weights of the edges incident with $v_i$ under some restrictions and, at the end of the step, we define $w(v_i)$ as the sum of the weights of the edges incident with $v_i$ at the end of  step $i$. The weights of edges must be in $\{1,\ldots, 7\}$.

In the $i$-th  step of ALGORITHM we treat vertex $v_i$. However, we merge the first and second steps of ALGORITHM, the vertices $v_1$ and $v_2$ are treated together.  Then, we  consider the remaining vertices according to the ordering ${\bf v}$. We assume $\om_1(e):=4$ for any $e\in E(G)$. Let $\om_2$ be the edge-weighting after the second step of ALGORITHM, $\om_i$ be the edge-weighting after treating the vertex $v_i$ (i.e. after $i$-th step of ALGORITHM),  and finally $\om:=\om_n$.  
	
	\medskip
	\noindent {\bf Step 1,2}
	
	We have $\sigma_{\om_1}(v_1)=4d_G(v_1)$ and $\sigma_{\om_1}(v_2)=4d_G(v_2)$. Observe that $4d \in \{0,2,4\} \pmod 6$ for every integer $d$. Let $e_1$ be  the edge between $v_1$ and its first successor distinct from $v_2$, let $e_2$ be the edge between $v_2$ and its first successor.  In Table \ref{table_0} we give the new weights of edges  $v_1v_2,e_1,e_2$.

	\begin{table}[!h]
		\begin{center}
			\begin{tabular}{|c|c|c|c|c|c|c|c|c|c|}
				\hline
				\begin{tabular}{cc}
					$(4d(v_1),4d(v_2))$\cr $\pmod 6$ \end{tabular} &(0,0)&(0,2)&(2,0)&(0,4)&(4,0)&(2,2)&(2,4)&(4,2)&(4,4) \cr
				\hline
				$\om_2(v_1,v_2)$& 7&7&7&5&5&5&6&6&2\cr
				\hline
				$\om_2(e_1)$&1&1&1&3&1&3&2&4&4 \cr
				\hline
				$\om_2(e_2)$&2&1&1&1&3&1&4&2&3 \cr
				\hline
			\end{tabular}
			
		\end{center}
		\caption{\label{table_0} Step 1,2 of ALGORITHM.}
	\end{table}
	
	\noindent We then put $w(v_1):=\sigma_{\om_2}(v_1),w(v_2):=\sigma_{\om_2}(v_2)$.

	Observe that after the first and the second steps of ALGORITHM, the vertex-colouring $w$ and the edge-weighting $\om_2$ have the following properties.
	
	\begin{obs}\label{obs1}
	
	\begin{itemize}
	\item $\sigma_{\om_2}(v_1),\sigma_{\om_2}(v_2)\in \{0,1,2\} \pmod 6$,
	\item $w(v_1)\neq w(v_2)$, namely $w(v_1)\not\equiv w(v_2)\pmod 6 $,
	\item the weight of the first successor of $v_i$ is at most 4 for $i\in\{1,2\}$.
	\end{itemize}
	\end{obs}
	\medskip
	\noindent {\bf Step $i,\;i\in\{3,\ldots,n\}$
	}
	
	Let $v_{k}$ be the first successor of $v_i$. For an edge $e$ the weight $w(e)$ can only be modified if either $e=v_jv_i$ with $j<i$ or $e=v_iv_{k}$. The weight of every edge must be in $\{1,\ldots, 7\}$. Furthermore, the modification of weights has to result in an edge-weighting $\om_i$ that satisfies the following properties:
	\begin{enumerate}[(1)]
		\item  $\om_i(v_iv_{k})\le 4$. 
		
		\item If $v_i\in V'$, then $\siom(v_i)\in \{0,1,2\} \pmod 6$. 
		
		\item Let $j< i$ and $v_j\in N(v_i)$.
		
		\begin{enumerate}[(i)]
			\item If $v_i\in V'$, then $\siom(v_i)\neq w(v_j)$.
			
			\item If $v_i\in V''$, then        
			\begin{itemize}
				\item if $v_j$ has no successor that follows $v_i$, then $\siom(v_i)\neq \siom(v_j)$;      
				\item if $v_j$ has a successor that follows $v_i$, then $\siom(v_i)\notin \{w(v_j),w(v_j)+3\}$.
			\end{itemize}
		\end{enumerate}
		
		\item  If $j<i$ and  $v_j\in N(v_i)$, then  $\siom(v_j)\in\{w(v_j),w(v_j)+3\}$.

		\item If $d(v_i)\ge 6$, then the set $\{v_jv_i:j< i,\; v_j\in N(v_i)\}$ is not monochromatic or the weight of  edges $\{v_jv_i:j< i,\; v_j\in N(v_i)\}$ is not in $\{\om_i(v_iv_{k}),\om_i(v_iv_{k})+3\}$.        
	\end{enumerate}
	
	\noindent When we obtain an edge-weighting that satisfies  properties (1)--(5), we assign $w(v_i):=\siom(v_i)$.
	
	\vspace{0.5cm}
	We will often use the following property of the vertex-colouring $w$ given by ALGORITHM:
	
	\begin{obs}\label{observation}
		If $u\in V'$, then $w(u)\in  \{0,1,2\} \pmod 6$.
	\end{obs}
	
	\begin{lemma}\label{executable}
		Every step $i\;(i\in \{3,\ldots,n\})$ of ALGORITHM is executable. 
	\end{lemma}
	
	\begin{proof}
		Let us consider the $i$-th step of ALGORITHM. We prove that we can modify the weights of edges between $v_i$ and its predecessors, and between $v_i$ and its first successor (if it exists), in such a way that we obtain an edge-weighting that satisfies  properties (1)--(5). We consider two cases, whether $v_i$ has a successor or not, each leading to several subcases.
		

		\medskip
		\noindent {\bf Case 1 $v_i\in V'$, \emph{i.e.}, $v_i$ has a successor}

		Let $v_{j_1},v_{j_2},\ldots,v_{j_d}$ be the predecessors of $v_i$ and $v_{j_{\ell}}v_i=e_{\ell}$ for $\ell \in\{1,\ldots, d\}$. Let $e'$ be the edge that joins $v_i$ with its first successor. Recall that $\om_{i-1}(e')=4$, $\om_{i-1}(e_{\ell})=4$ if $v_i$ is not the first successor of $v_{j_{\ell}}$, and $\om_{i-1}(e_{\ell})\le 4$ if $v_i$ is  the first successor of $v_{j_{\ell}}$.  We put the lower possible weights on every $e_{\ell}$ for $\ell \in\{1,\ldots, d\}$, \emph{i.e.}  we provisionally modify weights in the following way: $\om'_{i-1}(e_{\ell}):=\om_{i-1}(e_{\ell})-3$ if $\siomm(v_{j_{\ell}})=w(v_{j_{\ell}})+3$ and $\om'_{i-1}(e):=\om_{i-1}(e)$, otherwise. Observe that such a modification results in weights that belong to $\{1,\ldots ,4\}$, since $\om_{i-1}(e_{\ell})<4$ only if $v_i$ is the first successor of $v_{j_{\ell}}$ and then $\siomm(v_{j_{\ell}})=w(v_{j_{\ell}})$, otherwise $\om_{i-1}(e_{\ell})=4$.  To simplify the notations, we state $\om_{i-1}:=\om'_{i-1}$. After such a modification  we have $\siomm(v_{j_{\ell}})=w(v_{j_{\ell}})$ for $\ell \in\{1,\ldots, d\}$. 
		
		We will modify edges by adding 3 to $e_{\ell} $ for some $\ell \in\{1,\ldots, d\}$ or subtracting  1,~2 or 3 from the weight of $e'$.  As we observe above, by adding 3 to $e_{\ell} $, the weight of $e_{\ell}$ is still in $\{1,\ldots,7\}$. If we subtract 1, 2, or 3 from the weight of $e'$, then  the wight of $e'$ is  in $\{1,2,3\}$, since $\om_{i-1}(e')=4$.
		Furthermore, observe that adding 3 to some $e_{\ell} $  or subtracting  1, 2 or 3 from the weight of $e'$ maintains  properties (1) and (4). 
		We show now that, by such a modification of weights, we are able to result in the edge-weighting that also satisfies  properties (2),(3), and (5).
		
		By the modification weights of edges $e_1,\ldots, e_d,e'$, we see that $\siom(v_i)$ can take any value in the interval $[\alpha-3,\alpha -2,\ldots,\alpha+3d]$, where $\alpha=\siomm(v_i)$. 
		
		
		\medskip
		\noindent {\bf Subcase 1.1 $v_i$ has at least three predecessors}
		
		To satisfy the property (2), we have to choose weights for edges such that $\siom(v_i)\in \{0,1,2\} \pmod 6$, in the interval there are at least $d+3$ integers that are  congruent to 0, 1 or 2 $\pmod 6$. The property (3) can block at most $d$ values and hence $3$ values remain open for $\siom(v_i)$. Let $\beta_i\in [\alpha-3,\alpha -2,\ldots,\alpha+3d]\;(i\in\{1,2,3\})$ be the values open for $\siom(v_i)$, \emph{i.e.}   $\beta_i \in\{0,1,2\} \pmod 6$ and $\beta_i\neq\siom(v_{j_{\ell}})$ for all $\ell\in\{1,\ldots,d\}$.
		Let us denote $\beta_i=\alpha+3p_i-r_i$, where $p_i\in\{0,\ldots , d\}$ and $r_i\in\{0,1,2,3\}$ (\emph{i.e.} $p_i$ denotes the number of edges to which we have to add 3, $r_i$ denotes the value which we have to subtract from the weight of $e'$).  Now we have to guarantee the property (5).
		
		Suppose that there is $i$ such that $p_i\in\{1,\ldots, d-1\}$. We choose exactly $p_i$ edges from the set $\{e_1,\ldots,e_d\}$ and add 3 to their weights, next we subtract $r_i$ from the weight of $e'$. Since we choose $p_i$ edges from the set of $d$ edges and $0<p_i<d$, we can do this in such a way that  the property (5) holds. 
		
		Suppose that $p_i=0$ or $p_i=d$ for all $i\in\{1,2,3\}$. If edges $\{e_{\ell}:\ell\in\{1,\ldots ,d\}\}$ are not monochromatic, then every $\beta_i$ is good for $\siom(v_i)$. Thus,  we reweight only the edge $e'$ with $\om_i(e'):=\om_{i-1}(e')-r_1$, whenever $p_1=0$ or $\om_i(e_{\ell}):=\om_{i-1}(e_{\ell})+3$ for $\ell\in\{1,\ldots ,d\}$ and $\om_i(e'):=\om_{i-1}(e')-r_1$, otherwise. 
		
		Then assume that $p_i=0$ or $p_i=d$ and edges $\{e_{\ell}:\ell\in\{1,\ldots ,d\}\}$ are  monochromatic. 
		Thus  $\beta_i\in\{\alpha-3,\alpha-2,\alpha-1,\alpha+3d-2,\alpha+3d-1,\alpha+3d\}$ for all $i\in\{1,2,3\}$. There are at least two indexes $i$, say $i=1$ and $i=2$, such that  $\beta_1,\beta_2\in\{\alpha-3,\alpha-2,\alpha-1\}$ or $\beta_1,\beta_2\in\{\alpha+3d-2,\alpha+3d-1,\alpha+3d\}$ and so  we have two choices for the weight of $e'$. We can see that one of them results in an edge-weighting $\om_i$ such that the weight of  edges $\{e_{\ell}:\ell\in\{1,\ldots ,d\}\}$ is not in $\{\om_i(e'),\om_i(e')+3\}$ and hence the property (5) holds.  
		
		
		\medskip
		\noindent {\bf Subcase 1.2 $v_i$ has two predecessors}
		
		\begin{table}[!h]
			\begin{center}
				\begin{tabular}{||c||c|c|c|c|c|c|c|c|c|c||}
					\hline \hline
					&$\alpha-3$&$\alpha-2$&$\alpha-1$&$\alpha$&$\alpha+1$&$\alpha+2$&$\alpha+3$&$\alpha+4$&$\alpha+5$&$\alpha+6$\cr
					\hline \hline
					1&0&1&2&3&4&5&0&1&2&3 \cr
					\hline
					2&1&2&3&4&5&0&1&2&3&2 \cr
					\hline
					3&2&3&4&5&0&1&2&3&4&5 \cr
					\hline
					4&3&4&5&0&1&2&3&4&5&0 \cr
					\hline
					5&4&5&0&1&2&3&4&5&0&1 \cr
					\hline
					6&5&0&1&2&3&4&5&0&1&2 \cr
					\hline
				\end{tabular}
				
			\end{center}
			\caption{\label{table_1} Subcase 1.2, all possible values $\pmod 6$ in the interval.}
		\end{table}

		Thus, in the interval $[\alpha-3,\alpha -2,\ldots,\alpha+3d]=[\alpha-3,\alpha -2,\ldots,\alpha+6]$, there are at least $4$ integers that are  congruent to 0, 1 or 2 $\pmod 6$. The property (3) can block at most two values and hence two values remain open for $\siom(v_i)$. Let $\beta_i\in [\alpha-3,\alpha -2,\ldots,\alpha+6]\;(i\in\{1,2\})$ be the values open for $\siom(v_i)$. Similarly as above, let $\beta_i=\alpha + 3p_i-r_i$, where $p_i\in\{0,1,2\}$ and $r_i\in\{0,1,2,3\}$ for $i\in\{1,2\}$. 
		
		Suppose that either $p_1=1$ or $p_2=1$, say $p_1=1$. Then we add 3 to either $e_1$ or $e_2$ to obtain the edge-weighting such that $\om_i(e_1)\neq \om_i(e_2)$ and put $\om_i(e'):=\om_{i-1}(e')-r_1$. 
		
		Thus, we may assume that $p_i\in\{0,2\}$ and so $\beta_i\in\{\alpha-3,\alpha-2,\alpha-1,\alpha+4,\alpha+5,\alpha+6\}$ for all $i\in\{1,2\}$. If the edges $e_1$ and $e_2$ have different  weights, then  every $\beta_i$ is good for $\siom(v_i)$. Thus,  we reweight only the edge $e'$ with $\om_i(e'):=\om_{i-1}(e')-r_1$, whenever $p_1=0$ or $\om_i(e_1):=\om_{i-1}(e_1)+3,\om_i(e_2):=\om_{i-1}(e_2)+3$ and $\om_i(e'):=\om_{i-1}(e')-r_1$, otherwise.  
		
		Assume then that $p_i\in\{0,2\}$ and that $e_1$ and $e_2$ have the same weight. If  we have either $\beta_1,\beta_2\in\{\alpha-3,\alpha-2,\alpha-1\}$ or $\beta_1,\beta_2\in\{\alpha+4,\alpha+5,\alpha+6\}$, then we have two choices for the weight of $e'$. We can see that one of them gives an edge-weighting $\om_i$ such that the weight of  the edges $\{e_1,e_2\}$ is not in $\{\om_i(e'),\om_i(e')+3\}$ and hence the property (5) holds. 
		
		We claim that we always have either $\beta_1,\beta_2\in\{\alpha-3,\alpha-2,\alpha-1\}$ or $\beta_1,\beta_2\in\{\alpha+4,\alpha+5,\alpha+6\}$. 
		Let us consider the integers $\{\alpha,\alpha+1,\alpha+2,\alpha+3\}$, we can see that there is at least one value congruent to 0, 1 or $2 \pmod 6$ (see Table \ref{table_1}). We may assume that all values congruent to 0, 1 or $2 \pmod 6$ are blocked by the property (3), otherwise we are in the case considered above. Thus, we are not in the case described in lines 3 or 4 of Table \ref{table_1}.
		If there is exactly one value congruent to 0, 1 or $2 \pmod 6$ in $\{\alpha,\alpha+1,\alpha+2,\alpha+3\}$ (it is blocked by the property (3)), then there are five values congruent to 0, 1 or $2 \pmod 6$ in $\{\alpha-3,\alpha-2,\alpha-1,\alpha+4,\alpha+5,\alpha+6\}$ (see Table \ref{table_1}, lines 1 and 6),  at least four are not blocked by the property (2), and hence two of them are in either $\{\alpha-3,\alpha-2,\alpha-1\}$ or $\{\alpha+4,\alpha+5,\alpha+6\}$. 
		If there are two values congruent to 0, 1 or $2 \pmod 6$ in  $\{\alpha,\alpha+1,\alpha+2,\alpha+3\}$, then, there are three values congruent to 0, 1 or $2 \pmod 6$ in $\{\alpha-3,\alpha-2,\alpha-1,\alpha+4,\alpha+5,\alpha+6\}$ (see Table \ref{table_1}, lines 2 and 5) and none of them is blocked by the property (3) and hence two of them are in either $\{\alpha-3,\alpha-2,\alpha-1\}$ or $\{\alpha+4,\alpha+5,\alpha+6\}$.
		
		\medskip
		\noindent {\bf Subcase 1.3 $v_i$ has one predecessor}
		
		\begin{table}[!h]
			\begin{center}
				
				\begin{tabular}{||c||c|c|c|c|c|c|c||}
					\hline \hline
					&$\alpha-3$&$\alpha-2$&$\alpha-1$&$\alpha$&$\alpha+1$&$\alpha+2$&$\alpha+3$\cr
					\hline \hline
					1&0&1&2&3&4&5&0 \cr
					\hline
					2&1&2&3&4&5&0&1 \cr
					\hline
					3&2&3&4&5&0&1&2 \cr
					\hline
					4&3&4&5&0&1&2&3 \cr
					\hline
					5&4&5&0&1&2&3&4 \cr
					\hline
					6&5&0&1&2&3&4&5 \cr
					\hline
				\end{tabular}
				
			\end{center}
			\caption{\label{table_2} Subcase 1.3, all possible values $\pmod 6$ in the interval. }
		\end{table}
		
		Suppose first that $\alpha\notin\{0,1,2\} \pmod 6$. Then, there are at least four values congruent to 0, 1 or $2 \pmod 6$ in the interval $[\alpha-3,\alpha -2,\ldots,\alpha+3d]=[\alpha-3,\alpha -2,\ldots,\alpha+3]$ (see Table \ref{table_2}). One of them can be blocked by the property (3); so three values remain open for $\siom(v_i)$. Let $\beta_i\;(i\in\{1,2,3\})$ be the values open for $\siom(v_i)$. Thus at least two of them are in either $\{\alpha-3,\alpha-2,\alpha-1\}$ or $\{\alpha+1,\alpha+2,\alpha+3\}$, and so  we have two choices for the weight of $e'$. We can see that one of them gives an edge-weighting $\om_i$ such that  $\om_i(e_1)\notin\{\om_i(e'),\om_i(e')+3\}$, which guarantee that the property (5) holds. 
		
		Finally, suppose that $\alpha\in\{0,1,2\} \pmod 6$. Assume that there is $\beta_i$ such that $\beta_i=\alpha$ (\emph{i.e.} $\alpha$ is not blocked by the property (3) for $\siom(v_i)$). Recall that  $\om_{i-1}(e_1)\le 4$ and $\om_{i-1}(e')= 4$. If $\om_{i-1}(e_1)\neq 4$, then we assign $\om_i(e):=\om_{i-1}(e)$ for every $e\in E(G)$. If $\om_{i-1}(e_1)=4$, then we reweight edges $\om_{i}(e_1):=7$ and $\om_{i}(e'):=1$. 
		Suppose that $\alpha $ is blocked by  the property (3). If $\alpha\equiv 0 \pmod 6$, then there is a value congruent to 1 and there is a value congruent to 2 $\pmod 6$ in $\{\alpha+1,\alpha+2,\alpha+3\}$ (see Table \ref{table_1}, line 4) and hence one of them gives an edge-weighting $\om_i$ such that  $\om_i(e_1)\notin\{\om_i(e'),\om_i(e')+3\}$.  If $\alpha\equiv 2 \pmod 6$, then there is a value congruent to 0 and there is a value congruent to 1 $\pmod 6$ in $\{\alpha-3,\alpha-2,\alpha-1\}$ (see Table \ref{table_1}, line~6); so similarly as above we are done. If $\alpha\equiv 1 \pmod 6$, then $\beta_1=\alpha -1$ and $\beta_2=\alpha +1$ (see Table \ref{table_1}, line 5). If $\om_{i-1}(e_1)\neq 3$, then   we assign $\om_i(e'):=3$ and so $\siom(v_i)=\alpha -1$. Otherwise, we modify the  weights of two edges $\om_i(e_1):=6,\om_i(e'):=2$ and then $\siom(v_i)=\alpha+1$.

		
		\medskip
		\noindent {\bf Case 2 $v_i\in V''$, \emph{i.e.}, $v_i$ has no successor}
		
		Let $v_{j_1},v_{j_2},\ldots,v_{j_d}$ be the  neighbours of $v_i$ and $v_{j_{\ell}}v_i=e_{\ell}$ for $\ell \in\{1,\ldots, d\}$. Let $v_{j_{1}}$ be a vertex that has a successor in $\{v_{i+1},\ldots,v_n\}$ if such one exists. Recall that by our choice of the ordering of vertices {\bf v}, there is at most one such a vertex (Lemma \ref{lem:vertex_order} (iii)). To guarantee the property (3), we choose the weight of the edges incident with $v_i$ in such a way that $\siom(v_i)\neq \siom(v_{j_{\ell}})$ for $\ell \in\{2,\ldots, d\}$ and $\siom(v_i)\notin \{w(v_{j_1}),w(v_{j_1})+3\}$ even if $v_{j_1}$ has no successor in $\{v_{i+1},\ldots,v_n\}$.

		Similarly as in Case 1, we put the lower possible weights on every $e_{\ell}$ for $\ell \in\{1,\ldots, d\}$,  we provisionally modify the weights of edges in the following way: $\om'_{i-1}(e_{\ell}):=\om_{i-1}(e_{\ell})-3$ if $\siomm(v_{j_{\ell}})=w(v_{j_{\ell}})+3$, and $\om'_{i-1}(e):=\om_{i-1}(e)$ otherwise. Similarly as in Case 1, we can see that after such a modification, the weight of $e_{\ell}$ is in  $\{1,2,3,4\}$. To simplify notations, we state $\om_{i-1}=\om'_{i-1}$. Observe that $\siomm(v_{j_{\ell}})=w(v_{j_{\ell}})$ for $\ell \in\{1,\ldots, d\}$  and   $\siomm(v_{j_{\ell}})\in\{0,1,2\} \pmod 6$ for $\ell \in\{1,\ldots, d\}$ (every $v_{j_{\ell}}$ belongs to $V'$). 
		
		We will modify weights by adding 3 to $\om_{i-1}(e_{\ell})$ for some $\ell \in\{1,\ldots, d\}$. We can see that after adding 3 to the weight of $e_{\ell}$, the weight is still in $\{1,\ldots, 7\}$. Furthermore,  adding 3 to some $e_{\ell} $   maintains the property  (4). Since $v_i$ has no successor, properties (1) and (2)  hold. We prove that we can add 3 to some edges in such a way that  properties (3) and (5) will be satisfied. Let $\siomm(v_i)=\alpha$. 
		
		
		\medskip
		\noindent{\bf Subcase 2.1 $d(v_i)\ge 3$}
		
		Observe that if $\alpha\in \{0,1,2\} \pmod 6$, then $\alpha+3\notin\{0,1,2\} \pmod 6$. Thus, we consider two cases.
		
		\medskip
		\noindent{\bf Subcase 2.1.1 $\alpha\in\{0,1,2\} \pmod 6$}
		
		If $\alpha \neq \siomm(v_{j_1})$ and $\alpha\neq \siomm(v_{j_{\ell}})$ for $\ell\in\{2,\ldots,d\}$, then we assign $\om_i(e):=\om_{i-1}(e)$ for all $e\in E(G)$. Recall that $w(v_{j_1})=\siomm(v_{j_1})$ and so $\siom(v_{i})\neq w(v_{j_1})$. We also have $\siom(v_{i})\neq w(v_{j_1})+3$, since $\siom(v_{i})\in \{0,1,2\} \pmod 6$ and $w(v_{j_1})+3 \notin \{0,1,2\} \pmod 6$. Thus, $\om_i$ satisfies (3). If the edges incident with $v_i$ are not monochromatic or $d(v_i)\le 5$, then we are done. Otherwise, we reweight the edge $\om_i(e_1):=\om_{i-1}(e_1)+3$. Thus, $\siom(v_i)=\alpha+3$. Our assumption $\alpha \in\{0,1,2\} \pmod 6$ implies that $\alpha +3\notin\{0,1,2\} \pmod 6$ and consequently $\siom(v_{i})\neq \siom(v_{j_{\ell}})$ for $\ell\in\{2,\ldots,d\}$. Furthermore, $\alpha+3=\siom(v_i)\neq w(v_{j_1})$, since $\alpha+3\notin \{0,1,2\} \pmod 6$ and $w(v_{j_1}) \in \{0,1,2\} \pmod 6$. We also have $\siom(v_{i})\neq w(v_{j_1})+3$, since $w(v_{j_1})+3=\siomm(v_{j_1})+3\neq \alpha+3=\siom(v_{i})$
		Thus, we have a weighting $\om_i$ that satisfies  properties (1)--(5). 
		
		\medskip
		Assume now that $\alpha =\siomm(v_{j_1})$ or there is $\ell\in\{2,\ldots , d\}$ such that $\alpha=\siomm(v_{j_{\ell}})$. 
		
		Suppose first that $\alpha =\siomm(v_{j_1})$.  Assume that there are at least  two vertices $v_{j_a},v_{j_b}\in\{v_{j_2},\ldots v_{j_d}\}$ such that $\siomm (v_{j_a}) \neq \alpha +6,\siomm (v_{j_b}) \neq \alpha +6$. We assign $\om_i(e_1):=\om_{i-1}(e_1)+3,\om_i(e_a):=\om_{i-1}(e_a)+3,\om_i(e_b):=\om_{i-1}(e_b)+3$. Thus, $\siom(v_i)=\alpha +9$. We show that the property (3) holds. Since $\alpha+9\notin \{0,1,2\} \pmod 6$, we have $\siomm(v_{j_{\ell}})\neq \alpha +9$ for $\ell\in\{2,\ldots,d\}$ and so $\siom(v_i)\neq \siom(u)$ for $u\in \{v_{j_2},\ldots,v_{j_d}\}\setminus \{v_{j_a},v_{j_b}\}$. Our assumptions $\siomm (v_{j_a}) \neq \alpha +6,\siomm (v_{j_b}) \neq \alpha +6$ imply that $\siom (v_{j_a}) \neq \alpha +9=\siom(v_i),\siom (v_{j_b}) \neq \alpha +9=\siom(v_i)$. Now consider $v_{j_1}$. Since $w(v_{j_1})=\alpha $, we have $\siom(v_i)\neq w(v_{j_1})$ and $\siom(v_i)\neq w(v_{j_1})+3$. Thus, the edge-weighting $\om_i$ verifies the property (3). If $d(v_i)\le 5$ or edges incident with $v_i$ are not monochromatic, then we are done. Otherwise, if there is another vertex $v_{j_c}\in \{v_{j_2},\ldots v_{j_d}\}$ such that $\siomm (v_{j_c}) \neq \alpha +6$, then we can reweight edges in the following way: $\om_i(e_1):=\om_{i-1}(e_1)+3,\om_i(e_a):=\om_{i-1}(e_a)+3,\om_i(e_c):=\om_{i-1}(e_c)+3$. Thus, suppose that this is not the case: in $\{v_{j_2},\ldots v_{j_d}\}$, there are at most two vertices with colour other than $\alpha +6$.  Then, we add 3 to the weight of $e_1$ and edges incident with vertices with colours other that  $\alpha +6$.  Next,  from the remaining edges, we choose  one edge if we have two vertices with colours other that  $\alpha +6$, two edges if we have one vertex with colour other that  $\alpha +6$ and three edges if  we have no vertices with colour other that  $\alpha +6$, and add 3 to their weights.  Thus, we   obtain $\siom(v_i)=\alpha+12$. Since $d(v_i)\ge 6$,  we can choose edges for the reweighting in such a way that the edges incident with $v_i$ are not  monochromatic. Observe that  the only neighbours of $v_i$ that have in $\om_i$ the same colour as in $\om_{i-1}$ are those with colour $\alpha+6$. Those vertices are distinguished with $v_i$ in $\om_i$. Now, the remaining neighbours of $v_i$ have   colours that are not in $\{0,1,2\} \pmod 6$. Thus, they  are also distinguished from $v_i$ in $\om_i$, since $\alpha+12\in\{0,1,2\} \pmod 6$. So the edge-weighting $\om_i$ satisfies  properties  (1)--(5). 
		
		\medskip
		Finally, assume that $\alpha \neq\siomm(v_{j_1})$ and  there is $\ell\in\{2,\ldots , d\}$ such that $\alpha=\siomm(v_{j_{\ell}})$. 
		If the edges $\{e_2,\ldots,e_d\}$ are not monochromatic, or the weight of $\{e_2,\ldots,e_d\}$ is different from $\om_{i-1}(e_1)+3$, or $d(v_i)\le 5$, then we assign $\om_i(e_1):=\om_i(e_1)+3$. Since $\siom(v_i)=\alpha+3\notin \{0,1,2\} \pmod 6$, $v_i$ is distinguished from every vertex in $\{v_{j_2},\ldots v_{j_d}\}$. Our assumption $\alpha \neq\siomm(v_{j_1})$ implies $\siom(v_i)\neq w(v_{j_1})+3$. Furthermore, $\siom(v_i)\neq w(v_{j_1})$ since $\siom(v_i)\notin \{0,1,2\} \pmod 6$ and $w(v_{j_1})\in\{0,1,2\} \pmod 6$. Thus, the edge-weighting $\om_i$ verifies properties (1)--(5). 
		Thus, we may assume  that $d(v_i)\ge 6$ and $\om_{i-1}(e_2)=\ldots=\om_{i-1}(e_d)=\om_{i-1}(e_1)+3$. 
		
		If there is a $v_{j_a}\in\{v_{j_2},\ldots v_{j_d}\}$ with colour other than $\alpha$, then  we assign $\om_i(e_a):=\om_{i-1}(e_a)+3$. The edge-weighting $\om_i$ verifies properties (1)--(5) (recall that $\alpha+3\neq \siom(v_{j_{\ell}})$, since $\siom(v_{j_{\ell}})\in \{0,1,2\} \pmod 6$ for $\ell\in \{e_2,\ldots,e_d\}\setminus\{e_a\}$ and similarly as above we can observe that $\siom(v_i)\notin\{w(v_{j_1}),w(v_{j_1})+3\}$). 
		
		Suppose that all vertices $\{v_{j_2},\ldots v_{j_d}\}$ are coloured with $\alpha$. If $\alpha +6\neq w(v_{j_1})$, then we add 3 to the weights of two edges from $\{e_2,\ldots,e_d\}$. Since we can choose which edges to reweight, we can maintain the property (5). Since $\siom(v_i)=\alpha+6$, $\siom(u)=\alpha $ or $\alpha+3$ for $u\in \{v_{j_2},\ldots v_{j_d}\}$ and $\siom(v_i)\notin \{w(v_{j_1}),w(v_{j_1})+3\}$, $\om_i$ verifies properties (1)--(5).
		If $\alpha +6= w(v_{j_1})$, then  we add 3 to the weights of four edges. Again, we can choose which edges to reweight, since $d(v_i)\ge 6$. Hence, we are able to maintain the property (5). Similarly as above, we can check that $\om_i$ also verifies the property (3) and we are done.

		
		\medskip
		\noindent {\bf Subcase 2.1.2 $\alpha+3\in\{0,1,2\} \pmod 6$}

		Since $\alpha+3\in\{0,1,2\} \pmod 6$, we have $\alpha\notin\{0,1,2\} \pmod 6$ and, in $\{v_{j_1},\ldots v_{j_d}\}$, there is no vertex with colour $\alpha$ or  $\alpha +6$. 
		
		\medskip
		First, we consider the case when $\siomm(v_{j_1})=\alpha-3$. 
		
		If, in $\{v_{j_2},\ldots v_{j_d}\}$, there is a vertex $v_{j_a}$ with a colour other than $\alpha +3$, then we add 3 to  the weights of  $e_a$ and $e_1$. If the edges incident with $v_i$ are not monochromatic or $d(v_i)\le 5$, then we are done. Thus, suppose that $d(v_i)\ge 6$ and all these edges have the same weight. If there is  another vertex $v_{j_b}$, $b\neq a$, with a colour other than $\alpha +3$, then we  add 3 to the weights of $e_b$ and $e_1$. In the resulting edge-weighting, the edges incident with $v_i$ are not monochromatic. If $v_{j_a}$ is the only vertex with a colour other than $\alpha +3$, \emph{i.e.} all vertices in $\{v_{j_2},\ldots v_{j_d}\}\setminus\{v_{j_a}\}$ have the colour $\alpha +3$, then we choose one edge in $\{e_2,\ldots,e_d\}\setminus \{e_a\}$, say $e_b$, and assign $\om_i(e_1):=\om_{i-1}(e_1)+3,\om_i(e_a):=\om_{i-1}(e_a)+3,\om_i(e_b):=\om_{i-1}(e_b)+3$.  Since we have a choice,  we can maintain the property (5). 
		Now, we have 
		$\siom(v_i)=\alpha+9$,   
		$\siom(v_{b})=\alpha+6$, 
		$\siom(u)=\alpha+3$ for $u\in \{v_{j_2},\ldots v_{j_d}\}\setminus\{v_{j_a},v_{j_b}\}$; 
		so $\om_i$ distinguishes $v_i$ and vertices from $\{v_{j_1},\ldots v_{j_d}\}\setminus\{v_{j_b}\}$. 
		Furthermore, we have 
		$\siom(v_{a})=\siomm(v_{j_a})+3$. 
		As observed before, 
		$\siomm(v_{j_a})\neq \alpha +6$, 
		which implies that 
		$\siom(v_{a})\neq \siom(v_{i})$. 
		For  $v_{j_1}$ we have 
		$w(v_{j_1})=\siomm(v_{j_1})=\alpha-3$; so 
		$\siom(v_i)\notin \{w(v_{j_1}),w(v_{j_1})+3\}$. 
		Thus, the property (3) holds.

		If all vertices in $\{v_{j_2},\ldots v_{j_d}\}$ have colour $\alpha+3$, then  we choose three edges from $\{e_2,\ldots, e_d\}$ for the reweighting, and since we can choose freely, we can construct an edge-weighting $\om_i$ satisfying the property (5). Since
		$\siom(v_i)=\alpha+9$ and 
		$\siom(u)=\alpha+3$ or $\alpha+6$ for $u\in \{v_{j_2},\ldots v_{j_d}\}$, $\om_i$ distinguishes $v_i$ and vertices from $\{v_{j_1},\ldots v_{j_d}\}$.
		Similarly as above, we can see that  $\siom(v_i)\notin \{w(v_{j_1}),w(v_{j_1})+3\}$ and we are done. 
		
		\medskip
		Suppose that $\siomm(v_{j_1})=\alpha+3$. If the edges incident with $v_i$ are not monochromatic or $d(v_i)\le 5$, then the edge-weighting $\om_i:=\om_{i-1}$ satisfies (1)--(5). Thus, we may assume that all edges have the same weight  and $d(v_i)\ge 6$. 
		
		If, in $\{v_{j_2},\ldots v_{j_d}\}$, there are three  vertices $v_{j_a},v_{j_b},v_{j_c}$ with colour other than $\alpha +9$, then we add 3 to the weights of  $e_a,e_b,e_c$ and $e_1$. Thus, the edges incident with $v_i$ are not monochromatic (the property (5) holds) and  $\siom(v_i)=\alpha+12$. Since the colour of $v_{j_a},v_{j_b},v_{j_c}$ is not equal to $\alpha +9$ in $\om_{i-1}$, the colour of $v_{j_a},v_{j_b},v_{j_c}$ is not equal to $\alpha +12$ in $\om_{i}$. Thus, $\om_i$ distinguishes $v_i$ from $v_{j_a},v_{j_b},v_{j_c}$. Since  $\alpha+12\notin\{0,1,2\} \pmod 6$, no vertex in   $\{v_{j_2},\ldots v_{j_d}\}\setminus \{v_{j_a},v_{j_b},v_{j_c}\}$ has colour $\alpha +12$. Furthermore,  $w(v_{j_1})=\siomm(v_{j_1})=\alpha+3$; so $\siom(v_i)\notin \{w(v_{j_1}),w(v_{j_1})+3\}$ and hence  the  resulting edge-weighting satisfies properties (1)--(5).

		Assume that, in $\{v_{j_2},\ldots v_{j_d}\}$, there are at most two vertices with colour other than $\alpha +9$.
		Then, we add 3 to the weights of $e_1$ and edges incident with vertices having colour different from $\alpha +9$.  Next,  from the remaining edges, we choose  two edges if we have two vertices with colours other that  $\alpha +9$, three edges if we have one vertex with colour other that  $\alpha +9$ and four edges if  we have no vertex with colour other that  $\alpha +9$, and add 3 to their weights. Since $d(v_i)\ge 6$, we can choose the edges for  reweighting in such a way that the edges incident with $v_i$ are not monochromatic.   We   obtain $\siom(v_i)=\alpha+15$.  The vertices that had colour $\alpha +9$ in $\om_{i-1}$ have colour either $\alpha +9$ or $\alpha +12$ in $\om_{i}$; so $\om_i$ distinguishes $v_i$ and these vertices. Consider the vertices that had a colour different from $\alpha +9$ in $\om_{i-1}$. We added 3 to the edges incident with these vertices. In $\om_{i-1}$, the colours of these vertices were in $\{0,1,2\} \pmod 6$; so now these vertices have colours that are not in $\{0,1,2\} \pmod 6$, but $\siom(v_i)=\alpha +15\in \{0,1,2\} \pmod 6$. Thus, $\om_i$ distinguishes also these vertices. Furthermore,  $w(v_{j_1})=\siomm(v_{j_1})=\alpha+3$; so $\siom(v_i)\notin \{w(v_{j_1}),w(v_{j_1})+3\}$ and   we are done.

		\medskip
		Finally, suppose that $\siomm(v_{j_1})\notin\{\alpha-3,\alpha+3\}$. Since $\alpha\notin\{0,1,2\} \pmod 6$, there is no vertex with colour $\alpha$ in $\{v_{j_1},\ldots v_{j_d}\}$. If the edges incident with $v_i$ are not monochromatic or $d(v_i)\le 5$, then the edge-weighting satisfies properties (1)--(5). Thus, we may assume that all edges have the same weight  and $d(v_i)\ge 6$. 
		
		If, in $\{v_{j_2},\ldots v_{j_d}\}$, there is a vertex $v_{j_a}$ with colour other than $\alpha +3$, then we add 3 to  weights of  $e_a$ and $e_1$. Thus,  $\siom(v_i)=\alpha+6$ and so $\siom(v_{j_a})=\siomm(v_{j_a})+3\neq \siom(v_i)$, $\siom(v_i)\neq \siom(u)$ for $u\in \{v_{j_2},\ldots v_{j_d}\}\setminus \{v_{j_a}\}$, because $\siom(v_i)\notin\{0,1,2\} \pmod 6$ and $\siom(u)\in\{0,1,2\} \pmod 6$. Consider $v_{j_1}$: we have $\siom(v_i)\neq w(v_{j_1})$, since $ w(v_{j_1})=\siomm(v_{j_1})\in\{0,1,2\} \pmod 6$, and $\siom(v_i)\neq w(v_{j_1})+3$, since $\siomm(v_{j_1})\neq \alpha+3$ by our assumption. Thus, the resulting edge-weighting satisfies properties (1)--(5). 
		
		Thus, we may assume that $\siomm(v_{j_{\ell}})=\alpha+3$ for $\ell\in\{2,\ldots,d\}$. If $\siomm(v_{j_1})\neq\alpha+9$, then we choose three edges from $\{e_2,\ldots,e_d\}$ and add 3 to their weights. Since $d(v_i)\ge 6$, we can choose edges in such a way that we maintain the property (5). Now, we have  $\siom(v_{i})=\alpha+9$; so $\om_i$ distinguishes $v_i$ from $v_{j_{\ell}}$ for $\ell\in\{2,\ldots,d\}$. By our assumption, $\siom(v_{i})\neq w(v_{j_1})=\siomm(v_{j_1})$. Since    $ w(v_{j_1})+3\notin\{0,1,2\} \pmod 6$ and $\siom(v_i)\in\{0,1,2\} \pmod 6$, $\siom(v_{i})\neq w(v_{j_1})+3$ and we are done.
		If $\siomm(v_{j_1})=\alpha+9$, then we choose five edges from $\{e_1,\ldots,e_d\}$ and add 3 to their weights. Since $d(v_i)\ge 6$, we can choose edges in such a way that we maintain the property (5).  Now $\siom(v_{i})=\alpha+15$ and $\siom(u)=\alpha +3$ or $\alpha+6$ for $u\in \{v_{j_2},\ldots v_{j_d}\}$. Thus,   $v_i$ is distinguished from $\{v_{j_2},\ldots v_{j_d}\}$ by $\om_i$. By our assumption, $w(v_{j_1})=\alpha +9$ and so  the property (3ii) also holds.
		
		\medskip 
		\noindent {\bf Subcase 2.2 $d(v_i)=2$}
		
		Thus, $v_i$ has two neighbours $v_{j_1}, v_{j_2}$, and $v_{j_1}$ may have  a successor that follows $v_i$. Since $d(v_i)=2$, the property (5) holds.
		If $\alpha\neq \siomm(v_{j_2})$ and $\alpha \notin \{w(v_{j_1}),w(v_{j_1})+3\}$, then the edge-weighting $\om_i:=\om_{i-1}$ satisfies  properties (1)--(5). Thus, we may assume that either $\alpha=\siomm(v_{j_2})$ or $\alpha \in \{w(v_{j_1}),w(v_{j_1})+3\}$. 
		
		Suppose that $\alpha \in \{w(v_{j_1}),w(v_{j_1})+3\}$. First, assume that  $w(v_{j_1})=\alpha$ (\emph{i.e.} $\siomm(v_{j_1})= \alpha $). Thus, we must have $\alpha\in \{0,1,2\} \pmod 6$ and hence $\alpha+3\notin \{0,1,2\} \pmod 6$ which implies $\siomm(v_{j_2})\neq \alpha+3$. We assign $\om_i(e_1):=\om_{i-1}(e_1)+3$ and $\om_i(e_2):=\om_{i-1}(e_2)+3$ and we are done. Since now 
		$\siom(v_i)=\alpha+6\neq \siom(v_{j_2})$ 
		and $\siom(v_{i})\notin \{w(v_{j_1}),w(v_{j_1})+3\}$. 
		Suppose that $w(v_{j_1})+3=\alpha$. If $\siomm(v_{j_2})= \alpha+3$, then we assign $\om_i(e_2):=\om_{i-1}(e_2)+3$, otherwise, we assign $\om_i(e_1):=\om_{i-1}(e_1)+3$ and $\om_i(e_2):=\om_{i-1}(e_2)+3$. We can check that in both cases the property (3) holds.
		
		Thus, we may assume that  $\alpha \notin \{w(v_{j_1}),w(v_{j_1})+3\}$ and  $\siomm(v_{j_2})=\alpha$. The assumption $\siomm(v_{j_2})=\alpha$ implies  that $\alpha\in \{0,1,2\} \pmod 6$ and hence $\alpha+3\notin \{0,1,2\} \pmod 6$. In this case, we assign  $\om_i(e_1):=\om_{i-1}(e_1)+3$. Thus  $\alpha+3=\siom(v_{i})\neq \siom(v_{j_2})=\alpha$. Furthermore, 
		$\siom(v_{i})\neq w(v_{j_1})$, since 
		$\siom(v_{i})\notin \in\{0,1,2\} \pmod 6$ and 
		$w(v_{j_1})\in\{0,1,2\} \pmod 6$. Also
		$\siom(v_{i})\neq w(v_{j_1})+3$, since by our assumption  $\alpha \neq w(v_{j_1})$.

		\medskip 
		\noindent {\bf Subcase 2.3 $d(v_i)=1$}
		
		Thus $N(v_i)=\{v_{j_1}\}$, $v_{j_1}$ may have a successor that follows $v_i$, and $\siomm(v_i)=\om_{i-1}(v_{j_1}v_i)$.  Since $G\neq K_2$, we have $\siomm(v_i)<\siomm(v_{j_1})$ and so $\siomm(v_i)\notin \{w(v_{j_1}),w(v_{j_1})+3\}$. Thus, the edge-weighting $\om_i:=\om_{i-1}$ verifies properties (1)--(5). 
		
	\end{proof}
	
	
	\begin{lemma}\label{distinguishing}
		Let $\om$ be the edge-weighting given by ALGORITHM. Then $\om$ is a neighbour sum distinguishing $7$-edge-weighting.
	\end{lemma}
	
	\begin{proof}
		It is obvious that $\om$ is a $7$-edge-weighting, since the weight of every edge is in $\{1,\ldots,7\}$. We show that $\om$ is  neighbour sum distinguishing. Let ${\bf v},V',V''$ and $\om_i$ be defined the same as in ALGORITHM. Let $w$ be the vertex-colouring  determined by ALGORITHM. First, observe the following property of every vertex:
		
		\begin{claim}\label{claim}
			\begin{enumerate}[(i)]
				\item If $u\in V'$, then $\sigma_{\om}(u)\in \{w(u),w(u)+3\}$.
				
				\item If $u\in V''$, then $\sigma_{\om}(u)=w(u)$.
			\end{enumerate}
		\end{claim}
		
		\begin{proof}
			
			Let $u=v_i$.
			
			Suppose that $i=1$ or 2. Since $v_1$ and $v_2$ have successors, $v_1,v_2\in V'$. The values  $w(v_1)$ and $w(v_2)$ were assigned at the end of  steps 1 and 2, by $w(v_1)=\sigma_{\om_2}(v_1),w(v_2)=\sigma_{\om_2}(v_2)$. Observe that the weight of $v_1v_2$ will not change in steps $\{3,\ldots,n\}$. ALGORITHM has to respect the property (4); so the weights of the remaining edges incident with either $v_1$ or $v_2$ can be modified only in such a way that $\sigma_{\om_k}(v_1)\in \{w(v_1),w(v_1)+3\}$ and $\sigma_{\om_k}(v_2)\in \{w(v_2),w(v_2)+3\}$ for $k\in\{3,\ldots, n\}$. Thus, finally, $\sigma_{\om}(v_1)\in \{w(v_1),w(v_1)+3\}$ and $\sigma_{\om}(v_2)\in \{w(v_2),w(v_2)+3\}$.

			Suppose that $i\ge 3$. Assume first that $v_i\in V'$.  Let $v_{j_1},v_{j_2},\ldots,v_{j_d}$ be the  predecessors of $v_i$. Observe that the weight which we assigned to $v_{j_{\ell}}v_i$ in the $i$-th step of ALGORITHM will not change in the next steps (\emph{i.e.} $\om_i(v_{j_{\ell}}v_i)=\om(v_{j_{\ell}}v_i)$) and $w(v_i)=\sigma_{\om_i}(v_i)$. ALGORITHM has to verify the property (4); so the weights of edges  incident with the successors of $v_i$ can be modified only in such a way that $\sigma_{\om_k}(v_i)\in \{w(v_i),w(v_i)+3\}$ for $k\in\{i+1,\ldots, n\}$. Thus, finally, $\sigma_{\om}(v_i)\in \{w(v_i),w(v_i)+3\}$.

			Assume now that $v_i\in V''$.  Let $v_{j_1},v_{j_2},\ldots,v_{j_d}$ be the  neighbours of $v_i$. The weight which we assigned to $v_{j_{\ell}}v_i$ in the $i$-th step of ALGORITHM will not change in the next steps (\emph{i.e.} $\om_i(v_{j_{\ell}}v_i)=\om(v_{j_{\ell}}v_i)$). Thus  $\sigma_{\om_i}(v_i)=\sigma_{\om}(v_i)$, which implies that $w(v_i)=\sigma_{\om}(v_i)$. 
		\end{proof}
		
		Let $uw\in E(G)$, we show that $\sigma_{\om}(u)\neq \sigma_{\om}(w)$. Suppose that $uw=v_1v_2$. Steps~1~and~2 imply that $\{w(v_1),w(v_1)+3\}\cap \{w(v_1),w(v_1)+3\}=\emptyset$ and so $\sigma_{\om}(v_1)\neq \sigma_{\om}(v_2)$ by Claim \ref{claim}. Suppose that $u=v_j,w=v_i$ and $j<i\;(i\neq 2)$. We have ($v_i\in V'$ or $v_i\in V''$) and $v_j\in V'$, since $v_j$ has a successor. Suppose that $v_i\in V'$.  By the property (3i), $w(v_i)\neq w(v_j)$, since $w(v_i)=\siom(v_i)$. As we noticed in Observation~\ref{observation}, $w(v_i), w(v_j)\in\{0,1,2\} \pmod 6$, thus  $w(v_i)\neq w(v_j)+3$ and  $w(v_j)\neq w(v_i)+3$, and so $\{w(v_i),w(v_i)+3\}\cap \{w(v_j),w(v_j)+3\}=\emptyset$. Thus, Claim \ref{claim} implies that $\sigma_{\om}(v_i)\neq \sigma_{\om}(v_j)$. Suppose now that  $v_i\in V''$. Since $v_i$ has no successor, the weights of edges incident with $v_i$ will not change in steps $\{i+1,\ldots, n\}$, so $\siom(v_i)=\sigma_{\om}(v_i)$. If $v_j$ has no successor that follows $v_i$, then the weights of edges incident with $v_j$ will also not change in steps $\{i+1,\ldots, n\}$ and so $\siom(v_j)=\sigma_{\om}(v_j)$. By the property (3ii), we have  $\siom(v_i)\neq \siom(v_j)$, thus $\sigma_{\om}(v_i)\neq \sigma_{\om}(v_j)$ if $v_j$ has no successor that follows $v_i$. If $v_j$ has a successor that follows $v_i$, then the property (3ii) implies that $\siom(v_i)\notin \{w(v_j),w(v_j)+3\}$. As we observed, $\siom(v_i)=\sigma_{\om}(v_i)$ and then $\sigma_{\om}(v_i)\neq \sigma_{\om}(v_j)$ by Claim \ref{claim}. 
	\end{proof}

	\begin{lemma}\label{relaxed}
		Let $\om$ be the edge-weighting given by ALGORITHM. Then, for every vertex $u$ of degree at least $6$, there are edges $e'$ and $e''$ incident with $u$ satisfying $\om(e')\neq \om(e'')$.
	\end{lemma}
	
	\begin{proof}
		Let ${\bf v},V',V''$ and $\om_i$ be defined the same as in ALGORITHM. Let $w$ be the vertex-colouring  determined by ALGORITHM. 
		First, we prove that the lemma is true for $v_1$ and $v_2$. 
		Let $v_i$ be the first successor of $v_1$ different from $v_2$. Let $v_j$ be  the first successor of $v_2$. Let $e_1=v_1v_i,e_2=v_2v_j$.
		Steps~1~and~2 of ALGORITHM imply that $\om_2(v_1v_2)\notin \{\om_2(e_1),\om_2(e_1)+3\},\om_2(v_1v_2)\notin \{\om_2(e_2),\om_2(e_2)+3\}$. Observe that the weight of $v_1v_2$ will not change in steps $\{3,\ldots,n\}$, \emph{i.e.} $\om(v_1v_2)=\om_2(v_1v_2)$. When $v_i\;(v_j)$ is treated, then $\sigma_{\om_{i}}(v_1)=\sigma_{\om_{2}}(v_1)=w(v_1)\;(\sigma_{\om_{j}}(v_2)=\sigma_{\om_{2}}(v_2)=w(v_2))$, because $v_i\;(v_j)$ is the first successor. Thus, the weight of $e_1\;(e_2)$ can be modified only by adding $3$, because the property (4) must hold. So $\om(e_1)\in \{\om_2(e_1),\om_2(e_1)+3\}\;(\om(e_2)\in \{\om_2(e_2),\om_2(e_2)+3\})$. Thus, the argument that  $\om_2(v_1v_2)\notin \{\om_2(e_1),\om_2(e_1)+3\},\om_2(v_1v_2)\notin \{\om_2(e_2),\om_2(e_2)+3\}$ implies that $\om(v_1v_2)\neq \om(e_1)$ and $\om(v_1v_2)\neq \om(e_2)$, so we are done.
		
		Suppose that $u\in V'$ and $u\notin \{v_1,v_2\}$. Assume that $u=v_i$ and $v_k$ is the first successor of $v_i$. By the property (5) of ALGORITHM, the edges $\{v_jv_i:j< i,\; v_j\in N(v_i)\}$ are not monochromatic or the weight of the edges $\{v_jv_i:j< i,\; v_j\in N(v_i)\}$ is not in $\{\om_i(v_iv_{k}),\om_i(v_iv_{k})+3\}$. If the edges $\{v_jv_i:j< i,\; v_j\in N(v_i)\}$ are not monochromatic, then there are two edges $v_{j'}v_i$ and $v_{j''}v_i$ such that $\om(v_{j'})\neq \om(v_{j''})$ and we are done. Otherwise, observe that the weight of $v_iv_k$ can be modified only if $v_k$ is being treated. When $v_k$ is being treated, then $\sigma_{\om_{k-1}}(v_i)=\sigma_{\om_i}(v_i)=w(v_i)$, because $v_k$ is the first successor of $v_i$. Since ALGORITHM restricts the property (4), the weight of $v_iv_k$ can be modified only by adding 3. Thus, $\om(v_iv_k)$ is different from the weights of the edges incident with the predecessors of $v_i$.
		
		Suppose that $u\in V''$. Since $v_i$ has no successor and the property (5) of ALGORITHM must hold, the edges $\{v_jv_i:j< i,\; v_j\in N(v_i)\}$ are not monochromatic. Thus, there are two edges $v_{j'}v_i$ and $v_{j''}v_i$ such that $\om(v_{j'})\neq \om(v_{j''})$ and we are done. 
	\end{proof}

	\begin{proof}[of Theorem \ref{thm:seven_colours}]
		We may assume that $G$ is connected, since otherwise the theorem holds by induction on each component. The theorem is obviously true if $G=K_{1,n-1}$. Thus, we  assume that $G\neq K_{1,n-1}$. By Lemma \ref{lem:vertex_order}, there is an ordering ${\bf v}=(v_1,v_2,\ldots ,v_n)$ of vertices of $G$ that satisfies  conditions (i)--(iii). Thus, we can apply ALGORITHM on $G$. Let $\om$ be  the edge-weighting $\om$ given by  ALGORITHM. By Lemmas \ref{distinguishing} and \ref{relaxed},  $\om$   is a neighbour sum distinguishing $7$-edge-weighting and all the vertices of degree at least $6$ are incident with at least two edges of different weights, which proves the theorem. 
	\end{proof}
	
	
	\section{Bipartite graphs}
	\label{sec:bipartite}
	
	In this section, we show that every nice bipartite graph has a $6$-edge-weighting which distinguishes adjacent vertices and in which every vertex of degree at least 2 is incident with at least two edges of different weights. In order to prove this result, we apply  a result  obtained by Karo\'nski et al. in \cite{KaLu04}. They considered edge-weightings with elements of a group and proved the following theorem:
	
	\begin{theorem}{\rm \cite{KaLu04}}
		\label{thm:group}
		Let $\Gamma$ be a finite abelian group of odd order and let $G$ be a non-trivial $|\Gamma|$-colourable graph. Then, there is an edge-weighting of $G$ with  elements of $\Gamma$ such that the resulting vertex-colouring is  proper. 
	\end{theorem}
	
	Theorem \ref{thm:group} implies that if $k$ is odd and $G$ is   non-trivially $k$-vertex colourable, then $G$ admits a neighbour sum  distinguishing $k$-edge-weighting. Furthermore, the proof of Theorem \ref{thm:group} implies that if $U_1,\ldots, U_k,\;|U_i| >0,\;1\le i \le k$ are colour classes of $G$, then there is  a neighbour sum  distinguishing $k$-edge-weighting  $\om$  such that $\sigma_{\om}(v_i)= i \pmod k$ for every  $v_i\in U_i\;(1\le i \le k)$. For the purpose of further theorem, we need this property. However, we will use it only for bipartite graphs and for the edge-weighting with $\{1,2,3\}$. Thus, we restate Theorem \ref{thm:group}  and the proof for such a special case.

	\begin{theorem}{\rm \cite{KaLu04}}
		\label{thm:bipartite-proper}
		Let $G$ be a connected bipartite graph on at least three vertices with the vertex partition $(V_1,V_2)$. Then, $G$ admits a  neighbour sum  distinguishing $3$-edge-weighting. Moreover, there is  a neighbour sum  distinguishing $3$-edge-weighting $\om$ of $G$ such that $\sigma_{\om}(v_1)\neq \sigma_{\om}(v_2) \pmod 3$ for every  $v_1\in V_1$ and $v_2\in V_2$.
	\end{theorem}
	\begin{proof}
		Let $x\in V(G)$ and  $d(x)\ge 2$. Without loss of generality, assume that $x\in V_1$. Let $e_1=xv'_2,e_2=xv''_2$.  We start with the weight 3 on all edges, so $\sum_{e\in E(G)}\om(e)= 0 \pmod 3$. We now try to modify the weights of edges, maintaining the sum of edge weights congruent to 0 $\pmod 3$, until all vertices of $V_1\setminus \{x\}$ have colours congruent to 1 $\pmod 3$. To do that, for each vertex $v$ of $V_1\setminus \{x\}$, we consider a path from $v$ to $x$ and add alternately $1$ and $2$ to the values of the edges along this path. After such an operation, the colour of $v$ is $1\pmod 3$, the colour of $x$ is changed, and all the colours of the other vertices are unchanged. Now, the only vertex of $V_1$ which may have a colour different from 1 $\pmod 3$ is $x$, and all the vertices of $V_2$ still have a colour congruent to 0 $\pmod 3$. If the colour of $x$ is not congruent to 0 $\pmod 3$, we are done; if not, we can finish by reweighting edge $e_1$ on $c_1$ and $e_2$ on $c_2$, where $c_1,c_2\in\{1,2,3\}$ and  $c_1= \om(e_1)+2, c_2= \om(e_2)+2 \pmod 3$. Finally, we obtain the desired edge-weighting $\om$, because 
		\begin{itemize}
			\item either $\sigma_{\om}(v_1)= 1 \pmod 3$ for $v_1\in V_1$ and  $\sigma_{\om}(v_2)= 0$ for $v_2\in V_2\setminus\{v'_2,v''_2\}$, $\sigma_{\om}(v'_2)\in \{0,2\} \pmod 3$ and $\sigma_{\om}(v''_2)\in \{0,2\} \pmod 3$,
			\item or $\sigma_{\om}(v_1)= 1 \pmod 3$ for $v_1\in V_1\setminus \{x\}$, $\sigma_{\om}(x)= 2 \pmod 3$ and $\sigma_{\om}(v_2)= 0$ for $v_2\in V_2$.
		\end{itemize} 
	\end{proof}

	We can apply Theorem \ref{thm:bipartite-proper} for  our version of the neighbour sum  distinguishing edge-weighting. To prove the main result of this section, we need also the following lemma:
	
	\begin{lemma}\label{lem:colouring}
		If $G$ is  bipartite, then there is a $2$-edge-weighting of $G$ such that every vertex of degree at least 2 is incident with two edges with different weights. 
	\end{lemma}
	
	\begin{proof}
		We proceed by induction on the number of vertices. The lemma is true for bipartite graphs with one or two vertices. Assume that the lemma is true for every bipartite graph with less than $n$ vertices. Let $G$ be a bipartite graph with $n\ge 3$. If $G$ is not connected, then by induction there is a $2$-edge-weighting of every component of $G$ such that every vertex of degree at least 2 is incident with at least two edges of different weights and we are done. Assume that $G$ is connected and $v$ is a vertex of minimum degree.  Let $G'=G-v$ and $\om$ be a $2$-edge-weighting of $G'$ such that every vertex of degree at least 2 is incident with at least two edges weighted differently. 
		We extend $\om$ to all the edges of $G$. 
		
		\medskip
		First, assume that  $d_G(v)=1$, then there are two possibilities. Let $u$ be the neighbour of $v$. If $d_{G'}(u) \geq 2$, then, by induction hypothesis, $u$ is already incident with two edges weighted differently, hence we can label the edge $uv$ with any weight. Otherwise, we weight the edge $uv$ with the weight not used by the edge incident with $u$ in $G'$.
		
		\medskip
		Now, assume that $d_G(v)= 2$ and let $N(v)=\{u,w\}$.
		Suppose first that $v$ has a neighbour of degree at least 2 in $G'$, say $d_{G'}(u)\ge 2$. In this case, the edge $uv$ can be weighted with either 1 or 2, because  the vertex $u$ is already incident with two edges weighted differently in $G'$. So, we first weight the edge $vw$ in such a way that $w$ is incident with two edges of different weights, and next we weight $vu$ with the weight different from $\om(vw)$.

		Thus, we may assume that  $d_{G'}(u)=1$ and $d_{G'}(w)=1$. Observe  that if  $\om(uu_1)\neq \om(ww_1)$ (where $u_1,w_1$ is the neighbour in $G'$ of $u,w$, respectively), then we can extend the weighting on all the edges of $G$.  In such a case  we weight $vu$ with the weight $\om(ww_1)$ and $vw$ with the weight $\om(uu_1)$.  
		
		Thus, we may assume that $\om(uu_1)= \om(ww_1)$, say without loss of generality  $\om(uu_1)=\om(ww_1)=1$. We reweight some edges of $G'$.  If $u_1$  is incident with at least two edges weighted with 1, then we reweight the edge $uu_1$ with  2. In the new weighting of $G'$, every vertex of degree at least 2 is incident with at least two edges weighted differently and there are two neighbours of $v$ having  incident edges weighted differently; so as observed above we can extend the weighting to the desired edge-weighting of $G$.   Suppose that  $uu_1$ is the only edge incident with $u_1$  with weight  1, the remaining edges having weight 2. Let $u_2\in N(u_1)\setminus \{u\}$. If $u_2$ is incident with at least two edges weighted with 2, then we reweight the edge $u_1u_2$ with 1 and the edge $uu_1$ with  2. We obtain an edge-weighting of $G'$ in which every vertex of degree at least 2 is incident with at least two edges weighted differently and there are two neighbours of $v$ having  incident edges weighted differently, so we  are done. Otherwise, we repeat this reweighting process. Suppose that, after $k$ steps, we obtain a reweighted path $P=u_0,u_1,u_2,\ldots, u_k\;(u=u_0)$. Let $u_{k+1}\in N(u_k)\setminus \{u_{k-1}\}$. Since every  vertex $u_i\;(i\in\{1,\ldots, k-1\})$ is incident with exactly one edge weighted with $\om(u_{i-1}u_i)$ and there is no odd cycle in $G$, we have $u_{k+1}\notin V(P)\setminus \{u_0\}$ and $u_{k+1}\neq w$. Furthermore, $u_{k+1}\neq u$ because $d_{G'}(u)=1$. Thus, the reweighting process eventually ends, and we obtain an alternating path $P=u_0,u_1,u_2,\ldots, u_t$ ($u_0=u$). Every  vertex $u_i\;(i\in\{1,\ldots, t\})$ of $P$ has degree at least 2 in $G'$ and is incident with exactly one edge weighted with $\om(u_{i-1}u_i)$, while $u_t$ is incident with at least two edges weighted with $\om(u_{t-1}u_t)$. We can swap the weight of the edges of $P$, keeping an alternating path and obtaining a $2$-edge-weighting of $G'$ in which  every vertex of degree at least 2 is incident with at least two edges weighted differently and where two neighbours of $v$ have  incident edges weighted differently; so we can extend the weighting on the edges incident with $v$ in such a way that we  obtain the desired edge-weighting.
		
		\medskip
		Finally, assume that $d_G(v)> 2$. Since $v$ is a vertex of minimum degree, each neighbour of $v$ has degree at least 2 in $G'$. Thus, every neighbour of $v$ is incident with two edges weighted differently in $G'$. Hence, we can weight every edge incident with $v$ with either colour 1 or 2, ensuring that the edges incident with $v$ are not monochromatic. 
	\end{proof}
	
	\begin{theorem}
		\label{thm:bipartite-relaxed}
		Let $G$ be a nice bipartite graph. Then, there is a neighbour sum  distinguishing $6$-edge-weighting such that every vertex of degree at least 2 is incident with at least two edges with different weights.
	\end{theorem}

	This result can be restated the following way:
	
	\begin{reptheorem}{thm:bipartite-relaxed}
		Every nice bipartite graph  $G$ verifies $\chi'^{\Delta-1}_{\sum}(G)\le 6$.
	\end{reptheorem}
	
	\begin{proof}
		Let $(V_1,V_2)$ be  the vertex partition of $G$. By Theorem \ref{thm:bipartite-proper}, there is a neighbour sum  distinguishing $3$-edge-weighting $\om$ of $G$ such that $\sigma_{\om}(v_1)\not= \sigma_{\om}(v_2) \pmod 3$ for every  $v_1\in V_1$ and $v_2\in V_2$. Let $E_i=\{e\in E(G): \om(e)=i\}$ for $i\in\{1,2,3\}$. By Lemma \ref{lem:colouring}, every subgraph induced by $E_i$  can be weighted with two weights in such a way that every vertex of degree at least 2 is incident with at least two edges weighted differently. Thus, we reweight  the edges of $E_1$ with weights 1 and 4 in such a way that every vertex of degree at least 2 is incident with at least two edges weighted differently, and similarly we reweight the edges of $E_2$ with  2 and 5, and the edges of $E_3$ with 3 and 6. Let us denote by $\om'$ the resulting edge-weighting. Observe that $\sigma_{\om}(v)=\sigma_{\om'}(v) \pmod 3$. Thus, $\om'$ is  neighbour sum  distinguishing; so $\om'$ is the desired edge-weighting. 
	\end{proof}
	
	The following theorem was proved in \cite{LuQi11}:
	
	\begin{theorem}{\rm \cite{LuQi11}}
		\label{thm:bipartite_two_colours}
		Let $G$ be a connected bipartite graph on at least three vertices with vertex partition $(V_1,V_2)$. If $|V_1|$ is even, then,  there is  a neighbour sum  distinguishing $2$-edge-weighting $\om$ of $G$ such that $\sigma_{\om}(v_1)\not= \sigma_{\om}(v_2) \pmod 2$ for every  $v_1\in V_1$ and $v_2\in V_2$.
	\end{theorem}
	
	Thus, we can apply Theorem \ref{thm:bipartite_two_colours} and, similarly  as Theorem \ref{thm:bipartite-relaxed}, we can prove  the following result, which can again be restated:
	
	\begin{theorem}
		\label{thm:bipartite-relaxed_two_colours}
		Let $G$ be a connected bipartite graph on at least three vertices with  vertex partition $(V_1,V_2)$ and $|V_1|$ be even. Then, $G$ admits a neighbour sum  distinguishing $4$-edge-weighting such that every vertex of degree at least 2 is incident with at least two edges of different weights.
	\end{theorem}
	
	\begin{reptheorem}{thm:bipartite-relaxed_two_colours}
		Every nice bipartite graph $G$ with an even part verifies $\chi'^{\Delta-1}_{\sum}(G)\le 4$.
	\end{reptheorem}

\section{Conclusion}

In this paper, we studied neighbour sum distinguishing edge-weightings with constraints on the edges incident with vertices of large enough degree: there must be at least two edges with different colours. Those results can be rewritten in the neighbour sum distinguishing $d$-relaxed $k$-edge-colouring framework, which generalizes both the 1-2-3 Conjecture and its proper variant. A general conjecture for this framework is that every nice graph $G$ verifies $\chi'^{d}_{\sum}(G)\le \left\lceil \frac{\Delta(G)}{d} \right\rceil +2$. As such, a specific case is that every nice graph $G$ is conjectured to verify $\chi'^{\Delta-1}_{\sum}(G)\le 4$.

We extended the work previously done on subcubic graphs in~\cite{DaDuPeSi}, and proved that every nice graph $G$ verifies $\chi'^{\Delta-1}_{\sum}(G)\le 7$, and that the bound can be improved to~6 for nice bipartite graphs. We also showed that the bound of~4 of the conjecture holds for nice bipartite graphs with an even part.

Similarly to the 1-2-3 Conjecture and its proper variant, the $d$-relaxed general conjecture seems quite difficult to tackle head-on, hence why we considered it from the angle of local constraints allowing us to study the case $d=\Delta-1$. Hence, we propose some open problems that go further in this direction.

First, is it possible to improve the bound from~6 to~4 for nice bipartite graphs with both odd parts? Our bounds come from results on the 1-2-3 Conjecture and some recolouring to verify the condition, so another method might be necessary.

Then, how to improve the bound for nice graphs in general, in order to go closer to the conjectured bound of~4? Is some refinement of the methods we used (which are classical for this family of problems) enough, or would we need new ideas?

Finally, what of other values for $d$? In particular, is it possible to obtain general results for $d=\Delta-t$ with  $2\le t\le \Delta -2$? In practice, the neighbour sum  distinguishing $k$-edge-weighting such that every vertex of large enough degree is incident with at least $t+1$ edges of different weights implies that $\chi'^{d}_{\sum}(G)\le k$.

\section*{Acknowledgments}

We would like to thank \'Eric Duch\^ene and Aline Parreau for the fruitful discussion during the writing of this paper.

We would also like to thank the anonymous referee for their useful suggestions.


\end{document}